\documentclass[12pt]{article}
\usepackage{mathrsfs}
\usepackage{natbib,graphicx,setspace,lscape,longtable,epstopdf}
\usepackage{mathrsfs,amsmath,amsthm,amssymb,color}
\usepackage{natbib,epsfig,graphicx,pdfpages}
\usepackage{rotating}
\usepackage{float}
\usepackage{bm}
\usepackage{ulem}
\usepackage{CJK}
\usepackage{ragged2e}
\usepackage{hyperref}
\usepackage{multirow}
\usepackage{textcomp}
\usepackage{xr}
\usepackage{soul}

\bibpunct{(}{)}{;}{a}{,}{,}

\newtheorem{theorem}{Theorem}
\newtheorem{lemma}{Lemma}
\newtheorem{proposition}{Proposition}

\newtheorem{corollary}{Corollary}
\newcommand{\csection}[1]
    {\begin{center}
        \stepcounter{section}
        {\bf\large\arabic{section}. #1}
    \end{center}
}

\newcommand{\scsection}[1]
    {\begin{center}
        {\bf\large #1}
    \end{center}
}

\newcommand{\csubsection}[1]{
\begin{center}
\stepcounter{subsection}
{\it\arabic{section}.\arabic{subsection}. #1}
\end{center}
}

\newcommand{\scsubsection}[1]{
\begin{center}
\stepcounter{subsection}
{\it #1}
\end{center}
}

\def\black{\color{black}}

\def\beq{\begin{equation}}
\def\eeq{\end{equation}}
\def\beqr{\begin{eqnarray}}
\def\eeqr{\end{eqnarray}}
\def\beqrs{\begin{eqnarray*}}
\def\eeqrs{\end{eqnarray*}}
\def\bet{\begin{theorem}}
\def\eet{\end{theorem}}
\def\bel{\begin{lemma}}
\def\eel{\end{lemma}}
\def\bep{\begin{proposition}}
\def\eep{\end{proposition}}
\def\bg{\begin{figure}[tbph]\begin{center}}
\def\eg{\end{center}\end{figure}}

\def\bc{\begin{center}}
\def\ec{\end{center}}

\def\mR{\mathbb{R}}

\def\mF{\mathcal F}

\renewcommand{\arraystretch}{1.3}

\numberwithin{equation}{section}

\begin{document}
\begin{center}
{\bf\Large Do more observations bring more information in rare events?}\\
\bigskip

Danyang Huang$^{1}$, Liyuan Wang$^1$, and Liping Zhu$^{2,*}$

{\it\small
$^1$  Center for Applied Statistics, School of Statistics, Renmin University of China;\\
$^2$  Institute of Statistics and Big Data, Renmin University of China.


}

\end{center}

\begin{footnotetext}[1]
{Liping Zhu is the corresponding author, Institute of Statistics and Big Data, Renmin University of China, Beijing 100872, P.R. China. Email: zhu.liping@ruc.edu.cn. This work was supported by the National Natural Science Foundation of China (grant numbers 72471230, 12071477) and the MOE Project of Key Research Institute of Humanities and Social Sciences (22JJD110001), and the fund for building world-class universities (disciplines) at Renmin University of China. The authors gratefully acknowledge the support of Public Computing Cloud, Renmin University of China.}
\end{footnotetext}

\begin{singlespace}
\begin{abstract}

 It is generally believed that more observations provide more information. However, we observe that in the independence test for rare events, the power of the test is, surprisingly, determined by the number of rare events rather than the total sample size. Moreover, the correlations tend to shrink to zero even as the total sample size increases, as long as the proportion of rare events decreases. We demonstrate this phenomenon in both fixed and high-dimensional settings. To address these issues, we first rescale the covariances to account for the presence of rare events. We then propose a boosted procedure that uses only a small subset of non-rare events, yet achieves nearly the same power as using the full set of observations. As a result, computational complexity is significantly reduced. The theoretical properties, including asymptotic distribution and local power analysis, are carefully derived for both the rescaled statistic based on the full sample and the boosted test statistic based on subsampling. Furthermore, we extend the theory to multi-class rare events. Extensive simulations and real-world data analyses confirm the effectiveness and computational efficiency of the proposed approach.

\end{abstract}
\end{singlespace}

\newpage

\csection{INTRODUCTION}

Big data often contains rare events, also known as imbalanced data, where the number of occurrences (one class of the outcome) is significantly lower than the number of non-occurrences (the other class). Rare event data are prevalent across many scientific fields, including fraud detection \citep{phua2010comprehensive,chung2023credit}, medical diagnosis \citep{chawla2004special,haixiang2017learning,you2024sequential}, and customer churn prediction \citep{huang2013effective,verbeke2012building}. Following previous literature, we refer to events as {\it cases} and non-events as {\it controls}, with corresponding sample sizes denoted as $n_1$ and $n_0$. In rare event analysis, a key challenge is identifying features that depend on the class label through independence tests.

Independence testing plays a fundamental role in statistical analysis
\citep{hoeffding1948ANonParametric,gretton2007kernel,berrett2019nonparametric,li2020projective,shi2022distribution,deb2023multivariate,tong2023model}. Classical test statistics include Pearson correlation, Spearman's $\rho$ \citep{spearman1904proof}, and Kendall's $\tau$ rank correlation \citep{kendall1938new}. In recent years, independence tests have gained increasing attention, focusing on multivariate dependence testing and addressing high-dimensional data challenges \citep{runze2022linear,xu2024reducing}. Popular independence test statistics include, but are not limited to, distance correlation \citep{szekely2007measuring,sejdinovic2013equivalence,gao2021asymptotic}, projection correlation \citep{zhu2017projection,zhang2024projective}, and energy statistics \citep{szekely2013energy,deb2023multivariate}. These methods have demonstrated strong potential for capturing complex dependencies in high-dimensional settings. It is commonly believed that more data bring more information. Traditional theory suggests that larger sample sizes lead to better test performance. It is natural to ask whether this holds in the context of independence tests for rare events.


When the sample sizes of cases and controls are comparable, the test statistic under the null hypothesis depends on the total sample size $n = n_1 + n_0$. However, when $n_1/n_0 \to 0$, this conclusion becomes problematic. Consider a toy example where \( Y \) follows a Bernoulli distribution with parameter \( p_n \), and assume \( p_n \to 0 \) as the sample size \( n \) increases.
We evaluate the performance of Pearson correlation and Kendall's \(\tau\) rank correlation under two scenarios for the distribution of \( Y \). In both cases, when \( Y = 1 \), \( X \) is drawn i.i.d. from \( N(0.2,1) \), and when \( Y = 0 \), \( X \) is drawn from \( N(0,1) \).
\begin{figure}[htbp]
	\centering
\includegraphics[width=1\textwidth]{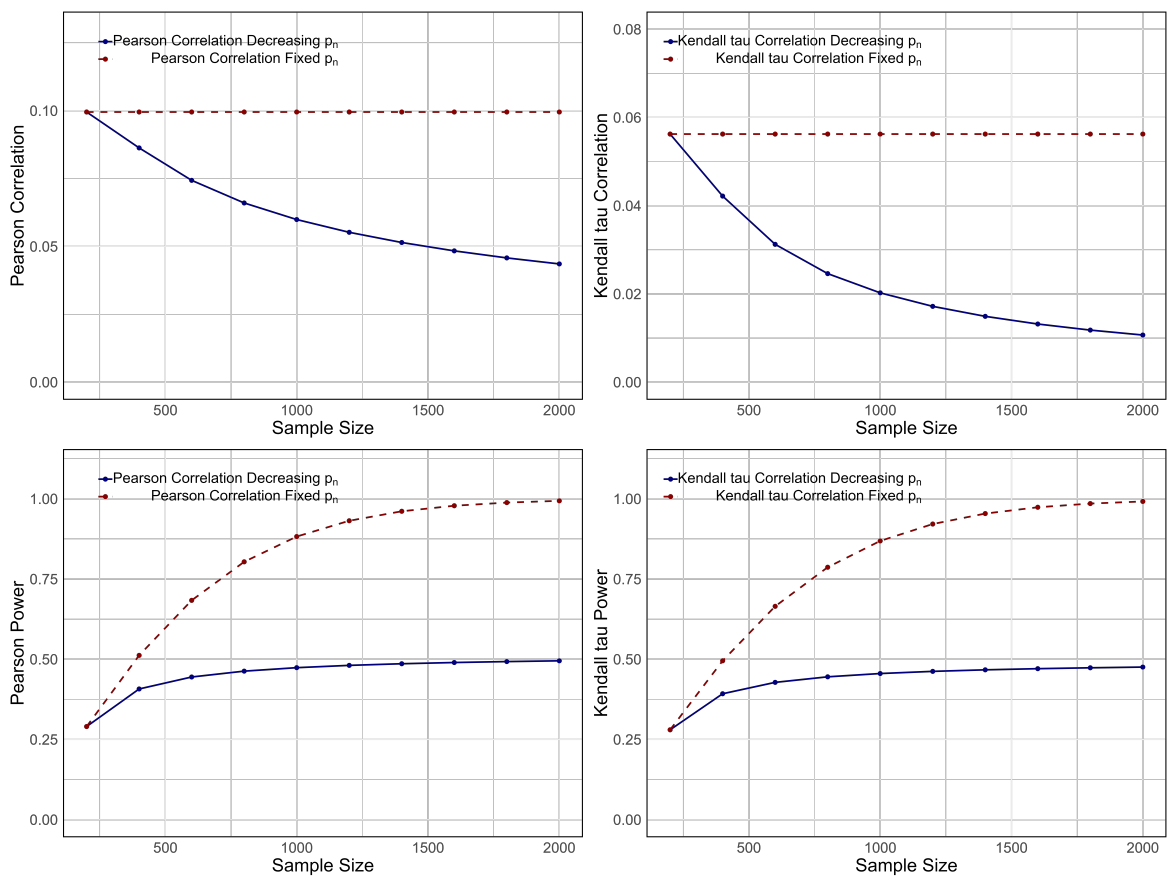}
	\caption{The upper row displays the expected correlation, while the lower row illustrates the test power. The left two plots show the results for Pearson correlation, while the right two plots correspond to Kendall's \(\tau\) rank correlation. The dashed lines represent the case with fixed \(p_n = 0.5\) (Scenario 1), whereas the solid lines correspond to decreasing \(p_n\) (Scenario 2).}\label{fig:rare_event}
\end{figure}

\noindent
Scenario 1. (Fixed $p_n$) For sample sizes ranging from 200 to 2,000, the marginal probability of $Y=1$ is fixed at $p_n=0.5$, which is balanced for all sample sizes.

\noindent
Scenario 2. (Decreasing $p_n$) With $n_1=100$ fixed, as $n$ increases from 200 to 2,000, the occurrence of $Y=1$ transitions to a rare event with $p_n$ decreasing from 0.5 to 0.025.

We examine the Pearson and Kendall \(\tau\) rank correlation and the test power, as shown in Figure \ref{fig:rare_event}. In Scenario 1, with fixed \(p_n\), the correlation remains stable as \(n\) grows, and the power of both tests rapidly approaches 1. In Scenario 2, as \(n\) increases, the correlation decreases, and the power stabilizes at a value significantly below 1 when the probability of cases decreases. This indicates that increasing the sample size does not enhance the identification of the correlation or reduce the Type II error as long as the proportion of cases decrease. The power seems to depend on the number of cases rather than the overall sample size, challenging classical theory. Thus, the key issue in rare event data is to rescale the statistic and analyze its rate of convergence.

Based on this phenomenon, an important question arises: can we achieve comparable test performance using a smaller subset of data proportional to the number of rare events? Rare event data inherently require sampling to mitigate computational challenges and improve estimation accuracy.
Two common importance sampling techniques could be used for imbalanced data: under-sampling controls \citep{drummond2003c4,liu2008exploratory} or over-sampling cases \citep{chawla2002smote,douzas2017improving}. For imbalanced data with \(c_1 < n_1/n_0 < c_2\) in logistic regression, \cite{fithian2014local} proposed a subsampling method assigning higher probability to controls similar to cases. In feature screening, \cite{xie2020fused} generalized Kendall's $\tau$ statistic using under-sampling controls and proposed a guaranteed screening procedure. For extremely imbalanced data with \(n_1/n_0 \rightarrow 0\), \cite{wang2020logistic} and \cite{wang2021nonuniform} showed that under-sampling achieves the same convergence rate as the full sample. \cite{song2024minimax} studied minimax risk bounds for imbalanced data, showing bounds depend on the effective sample size \( \kappa n \), where \( \kappa \) represents the maximum class imbalance. For high-dimensional linear discriminant analysis, \cite{mojiri2022new} examined the effects of imbalanced class sizes. Under-sampling also improves performance in machine learning methods, such as decision trees \citep{lin2017clustering}, support vector machines \citep{bao2016boosted}, and others \citep{spelmen2018review}.

Inspired by previous literature, this work focuses on the theoretical analysis of independence tests for rare events with \(n_1/n_0 \to 0\). We first rescale the test statistic and establish its asymptotic properties, demonstrating that the convergence rate is primarily determined by the number of rare events \(n_1\), rather than the total sample size \(n\). Many classical test statistics exhibit similar behavior and can be rescaled accordingly. To reduce computational costs in large-scale data, we propose a Boosted Independence Test (BIT) procedure that uses only a small subset of non-rare events, achieving nearly the same power as using the full sample. We investigate and compare the theoretical properties of both the rescaled statistics and the newly proposed one, considering both fixed and high-dimensional settings. Finally, we extend our results to accommodate multiple rare events, broadening the applicability of the framework.

The rest of the paper is organized as follows. In Section 2, we introduce the rescaled independence test for rare event scenarios and provide its asymptotic properties. Section 3 presents the results of simulation studies and real data analyses that demonstrate the effectiveness of the proposed method. Finally, in Section 4, we conclude with remarks and discuss potential future research directions.

\csection{BOOSTED INDEPENDENCE TEST FOR RARE EVENT}
\csubsection{Notations and Background}

Let \((X_i, Y_i)\) be the observation from the \(i\)-th subject, where \(1 \leq i \leq n\). Here, \(Y_i \in \{0,1\}\) is the class label and \(X_i = (X_{i1}, \dots, X_{ip})^\top \in \mR^p\) is the associated \(p\)-dimensional feature vector. We begin with binary \(Y\) and later extend to multi-category cases.
For the independence test, the null hypothesis is that the distribution of \(X\) is independent of \(Y\). Many correlation measures for balanced data can be framed within a U-statistic framework \citep{serfling2009approximation,li2020distributed}, which is given by,
\begin{equation}
	\label{eq:U_stat}
	U =  \binom{n}{m}^{-1}  \sum_{1 \leqslant i_1 < i_2 < \cdots < i_m \leqslant n} h(X_{i_1}, X_{i_2}, \ldots, X_{i_m}),
\end{equation}
where $m$ is the number of observations in each combination, and $h$ is a $m$-ary symmetric kernel function with zero mean. U-statistics are well-established theoretically \citep{hoeffding1948AClass, brown1978reduced, randles1979introduction, nolan1987u, huskova1993consistency, AOP2013, lee2019u}.
 By \citep{serfling2009approximation}, we define the $c$-observation projection as $h_c( X_{1},\cdots,X_{c} )=\mathrm{E}\{ h(X_{1}, X_{2}, \ldots, X_{m}) \mid  X_{1},\cdots,X_{c}\}$
for $1 \leqslant c<m $. The asymptotic distributions for $c$=1 and $c$=2 are established in Lemma \ref{lemma_tradition_U}.
\begin{lemma}
	\label{lemma_tradition_U}
	(1) \citep{hoeffding1948AClass} If $\xi_1 = \operatorname{Var} \big\{ h_1(X_1)\big\}  > 0$, then the U-statistic satisfies
	$n^{1/2} U \xrightarrow{d} N(0,m^2 \xi_1),$
	where $\xrightarrow{d}$ denotes convergence in distribution.

	(2) \citep{GREGORY1977} If $\xi_1 = 0$ but $\xi_2 =\operatorname{Var} \big\{ h_2(X_1,X_2)\big\}  > 0$, then we have
	$n U \xrightarrow{d} m(m-1)\eta_0/2,$
	where $\eta_0=\sum_{j=1}^{\infty} \lambda_j\left(\chi_{1 j}^2-1\right)$ and $\chi_{1j}^2$s are independent $\chi_1^2$ variables, and $\lambda_j$ is eigenvalues of the linear operator $A$, with $A g(x) = \int_{-\infty}^{\infty} h_{2}(x,y) g(y) dF(y)$  mapping a function $g(x)$ to another one $A g(x)$.
\end{lemma}

\noindent
Considering Kendall's $\tau$ rank correlation in the toy example from the introduction, it can be verified that,
\begin{equation}
    \label{eq:tau_origion}
    \hat{\tau} = \frac{2}{n(n-1)} \sum_{1 \leqslant i < j \leqslant n} \text{sgn}(X_i - X_j) \text{sgn}(Y_i - Y_j),
\end{equation}
where the kernel function is \( h\{(X_i, Y_i), (X_j, Y_j)\} = \text{sgn}(X_i - X_j) \text{sgn}(Y_i - Y_j) \). For fixed \(p_n = p_0\), if \( X \) and \( Y \) are independent, the asymptotic distribution is \(n^{1/2} \hat{\tau} \xrightarrow{d} N(0, 4 p_0 (1 - p_0) / 3)\). When \( X \) and \( Y \) are  dependent, \(\mathrm{E} \hat{\tau} =2\{2 \mbox{P}(X_1 > X_2 \mid Y_1 = 1, Y_2 = 0) - 1\} p_0 (1 - p_0) \). For \( Y_i = Y_j \), \( \operatorname{sgn}(Y_i - Y_j) = 0 \), which leads to the number of non-zero terms being \( n_0 n_1 \), so \( \left|\sum \operatorname{sgn}(X_i - X_j) \operatorname{sgn}(Y_i - Y_j)\right| \leqslant n_0 n_1 \). Then, for extremely imbalanced data with \( p_n \to 0 \), \( n^{1/2} \hat{\tau} \xrightarrow{p} 0 \) under \( \mbox{H}_0 \), and \( \mathrm{E} \hat{\tau} \to 0 \) under \( \mbox{H}_1 \), making the test statistic ineffective. This work focuses on the independent test in this extremely imbalanced data.

 The kernel's asymmetry makes projection-based methods in classical U-statistics infeasible. \citet{sen1974weak} used Brownian motion to establish the asymptotic normality of non-degenerate generalized U-statistics, later extended by \citet{neuhaus1977functional} to degenerate cases. \citet{lee2019u} introduced generalized U-statistics without comprehensive theoretical analysis. Most prior research focuses on balanced scenarios, leaving imbalanced cases largely unexplored.


\csubsection{Rescaled Independence Test for Rare Event}

For the theoretical analysis of the test statistic in an imbalanced dataset, we treat \( Y = 0 \) and \( Y = 1 \) as two distinct types of observations: \( \mathcal{D}^0 = \{ X_1^{(0)}, X_2^{(0)}, \dots, X_{n_0}^{(0)} \} \) for \( Y = 0 \), and \( \mathcal{D}^1 = \{ X_1^{(1)}, X_2^{(1)}, \dots, X_{n_1}^{(1)} \} \) for \( Y = 1 \). Inspired by the concept of generalized U-statistics \citep{sen1974weak,neuhaus1977functional,lee2019u}, we define the following statistic to estimate the covariance between \( X \) and \( Y \) for rare events:
\begin{equation}
	\label{eq:k2_U_form}
	T = \binom{n_0}{m_0}^{-1} \binom{n_1}{m_1}^{-1} \sum_c h\left(X_{i^0_{1}}^{(0)}, \cdots,X_{i^0_{m_0}}^{(0)} ;X_{i^1_{1}}^{(1)}, \cdots,X_{i^1_{m_1}}^{(1)}\right),
\end{equation}
where $h(\cdot)$ is assumed to be zero-mean and symmetric within each $\{X_{i^k_{1}}^{(k)}, \cdots, X_{i^k_{m_k}}^{(k)}\}$ for $k = 0, 1$; $\left\{i^k_{1}, \ldots, i^k_{m_k}\right\}$ denotes a set of $m_k$ distinct elements from the index set $\left\{1, 2, \ldots, n_k\right\}$; and $\sum_c$ denotes the summation over all possible combinations. We refer to the statistic in \eqref{eq:k2_U_form} as a rescaled independence test statistic for rare events, which is referred to as RIT for convenience.
 Note that the number of (non-zero) terms in the sum equals the denominator in the form of \eqref{eq:k2_U_form}. Therefore, if $X$ and $Y$ are not independent, $T$ can be verified to converge to a non-zero constant as the number of cases increases. For example, the test statistic generalized from Kendall's \( \tau \) can be defined with \( h(X^{(0)}, X^{(1)}) = \operatorname{sgn}(X^{(1)} - X^{(0)}) \) as,
\begin{equation}
 		\label{eq:tau_with_U_form}
 		T_{\tau}= \frac1{n_0 n_1} \sum_{i=1}^{n_1} \sum_{j=1}^{n_0}  \operatorname{sgn}\left( X_{i}^{(1)} - X_{j}^{(0)}\right) .
\end{equation}
\noindent

To establish the statistical properties of the test statistic in (\ref{eq:k2_U_form}), we define the $\left(a, b\right)$-observation projection as
\begin{equation}
	\label{eq:projection_h_ab}
	\begin{aligned}
		&h_{a,b} \left( X_{1}^{(0)} , \cdots,X_{a}^{(0)};X_{1}^{(1)} , \cdots,X_{b}^{(1)}\right)  \\
		= & \mathrm{E}\left\lbrace  h\left(X_{1}^{(0)}, \cdots,X_{m_0}^{(0)} ;X_{1}^{(1)}, \cdots,X_{m_1}^{(1)}\right) \mid  X_{1}^{(0)} , \cdots,X_{a}^{(0)}; X_{1}^{(1)} , \cdots,X_{b}^{(1)} \right\rbrace,
	\end{aligned}
\end{equation}
for $0 \leqslant a < m_0$ and $0 \leqslant b < m_1$. Define $\xi_{a,b} = \operatorname{Var} \{ h_{a,b} ( X_{1}^{(0)}, \cdots, X_{a}^{(0)}; X_{1}^{(1)}, \cdots, X_{b}^{(1)} ) \}$. Then we then investigate the theoretical properties of $T$ in the scenario of rare events. 

\begin{theorem} (The asymptotic properties of $T$)
    \label{thm:gen_U}
    (1) Under the null hypothesis and the condition of rare events $n_1/n \rightarrow 0$, if $\xi_{0,1} = \operatorname{Var} \{h_{0,1}( X_1^{(1)} ) \} > 0$, then the RIT satisfies
    $n_1^{1/2} T \xrightarrow{d} N\left( 0, {m_1}^2 \xi_{0,1} \right)$.
    (2) Assume $\xi_{0,1} = 0$ and $\xi_{0,2} > 0$. If at least one of the following conditions holds: (a) $\xi_{1,0} = 0$, or (b) $n_1^2 / n_0 \rightarrow 0$, then we have
   $ n_1 T \xrightarrow{d} m_1(m_1 - 1) \eta / 2$,
    where $\eta = \sum_{j=1}^{\infty} \lambda_j\left( \chi_{1 j}^2 - 1 \right)$, $\chi_{1j}^2$ are independent $\chi_1^2$ variables, and $\lambda_j$ are weights depending on the distribution of $X^{(1)}$ and $h_{0,2}(\cdot)$.
\end{theorem}

\noindent
See Appendix C for the detailed proof of Theorem \ref{thm:gen_U}. Notably, Theorem \ref{thm:gen_U} indicates that in rare event analysis, the convergence rate of the distribution of $T$ is determined solely by \(n_1\). This conclusion differs entirely from that obtained in Lemma \ref{lemma_tradition_U} for balanced data.
Note that $h_{0,2}(X_1^{(1)}, X_2^{(1)})$ differs from $h_2(X_1, X_2)$ in Lemma \ref{lemma_tradition_U}, as it is a projection onto the observations of cases, with the information from controls integrated out.  Thus, \( h_{0,2}(X_1^{(1)}, X_2^{(1)}) \) is a function that captures information from both cases and controls. Additionally, while Theorem \ref{thm:gen_U} provides the asymptotic distribution of \( n_1 T \) under the null hypothesis, this distribution is intractable due to the unknown weights \( \lambda_j \), which cannot be easily computed from the distribution of \( X \). Therefore, the random permutation method is typically used to approximate this distribution in independence tests \citep{berrett2020conditional,zhang2024projective}.

Theorem \ref{thm:gen_U} shows that the available information is determined by the scale of \( n_1 \), and a large sample size \( n \) does not necessarily imply a large amount of information. This is a key conclusion in rare event scenarios, emphasizing the importance of discussing the boosted independence test theory. Specifically, when \( \xi_{0,1} \neq 0 \), we refer to this test statistic as a {\it first-order RIT}. When \( \xi_{0,1} = \xi_{1,0} = 0 \) and \( \xi_{0,2} > 0 \), we require the imbalance to satisfy \( n_1 / n_0 \to 0 \). When \( \xi_{0,1} = 0 \), \( \xi_{1,0} \neq 0 \), and \( \xi_{0,2} > 0 \), a rarer condition \( n_1^2 / n_0 \to 0 \) is required.
We define the test statistic in these two cases as a {\it second-order RIT}. However, when \( \xi_{0,1} = 0 \), \( \xi_{1,0} \neq 0 \), \( \xi_{0,2} > 0 \), and \( n_1^2 / n_0 \nrightarrow 0 \), the asymptotic variance of \( T \) could still be influenced by the term \( n_0^{-1} {m_0}^2 \xi_{1,0} \). For simplicity, this case is not included in Theorem \ref{thm:gen_U}.

Note that determining the asymptotic distribution of the second-order RIT is computationally challenging. We then examine a special case where \( X \) is high-dimensional and refer to this test statistic as \( T_{\infty} \).
Further define $G(x, y) = \mathrm{E} \{ h_{0,2}(X^{(1)}, x) h_{0,2}(X^{(1)}, y) \}$. Then, we have the following theorem.

\begin{theorem} (The asymptotic properties of $T_\infty$)
    \label{thm:T_div_p}
    Under the null hypothesis and $n_1/n\rightarrow 0$, if $\xi_{0,1} = \xi_{1,0} = 0$, $p \rightarrow \infty$, and $h_{0,2}(\cdot, \cdot)$ depending on $n_1$ satisfying
    \begin{equation}
        \label{eq:condition_clt}
        \frac{\mathrm{E}\left\lbrace G^2\left( X^{(1)}_1, X^{(1)}_2 \right) \right\rbrace + n_1^{-1} \mathrm{E}\left\lbrace h_{0,2}^4\left( X^{(1)}_1, X^{(1)}_2 \right) \right\rbrace}{\mathrm{E}\left\lbrace h_{0,2}^2\left( X^{(1)}_1, X^{(1)}_2 \right) \right\rbrace^2} \rightarrow 0,
    \end{equation}
    then $n_1 T_\infty / \xi_{0,2}^{1/2}$ is asymptotically normally distributed with zero mean and variance $m_1^2 (m_1 - 1)^2 / 2$ as $n_1 \to \infty$.
\end{theorem}

\noindent
The proof of Theorem \ref{thm:T_div_p} is given in Appendix C. Notably, for second-order RIT, we only consider the projections of the kernel function $h_{0,2}(X_i^{(1)}, X_j^{(1)})$. Define $Z_i = \sum_{j=1}^{i-1} h_{0,2}(X_i^{(1)}, X_j^{(1)})$, $V = \sum_{i=2}^{n_1} Z_i$, and $s^2 = \mathrm{E}(V^2)$.
Condition (\ref{eq:condition_clt}) implies the following two conditions required by the {Lyapunov} martingale central limit theorem:
$s^{-4} \sum_{i=2}^{n_1} \mathrm{E}(Z_i^4) \to 0$, and $s^{-2} \sum_{i=2}^{n_1} \mathrm{E}(Z_i^2 \mid X_1^{(1)}, \ldots, X_{i-1}^{(1)}) \xrightarrow{p} 1$.
 Therefore, the asymptotic normality of $V$ follows from the martingale central limit theorem, which subsequently establishes the asymptotic normality of $T_\infty$. Notably, when \( p \) is fixed, condition (\ref{eq:condition_clt}) converges to a non-zero constant, causing the theorem to no longer hold.

To make valid statistical inferences based on (\ref{eq:k2_U_form}), one must estimate \( \xi_{0,1} \) or \( \xi_{0,2} \). If \( \xi_{0,1} = \operatorname{Var} \{ h_{0,1} ( X_1^{(1)} ) \} > 0 \), for the kernel function in \eqref{eq:k2_U_form}, the computational complexity of calculating \( \hat{\xi}_{0,1} \) is \( O(n^{m_1 + m_0}) \). By Slutsky's theorem, a consistent estimate of \( \hat{\xi}_{0,1} \) is sufficient, and thus a computationally efficient estimator can be obtained through subsampling.

%
%
%

\csubsection{Discussion and Rescaling of Test Statistics in the RIT Framework}

To provide a more detailed analysis, we discuss five examples. For the first-order RIT, we consider those rescaled from Pearson correlation, Kendall's \( \tau \) rank correlation, and a newly designed imbalanced Kendall's \( \tau \) statistic. For the second-order RIT, we examine those generated from distance correlation and improved projection correlation. All related verifications about kernel functions are provided in Appendix B.

1. Rescaled Pearson $T_r$.
Consider the kernel function $h\left( X^{(0)}, X^{(1)} \right) = X^{(1)} - X^{(0)}$. Then the rescaled Pearson $T_r$ is defined as
$  T_r = n_1^{-1} \sum_{i=1}^{n_1} X_i^{(1)} - n_0^{-1} \sum_{j=1}^{n_0} X_i^{(0)}$.
Note that $\xi_{0,1} > 0$.
Define $\sigma_{\rm x}^2 = \operatorname{Var}(X)$. Under the null hypothesis and $n_1/n\rightarrow 0$, we have
${n_1}^{1/2} T_r \xrightarrow{d} N(0, \sigma_{\rm x}^2)$ as $n_1 \to \infty.$
The proof of this conclusion requires only the calculation of the variance of the conditional expectation of the kernel function $\xi_{0,1}$, which is $\operatorname{Var}(X)$ in this case.

2. Rescaled Kendall's $T_\tau$.
We have discussed the rescaled Kendall's $\tau$ with $T_\tau$ defined in \eqref{eq:tau_with_U_form}. The kernel function is $h(X^{(0)}, X^{(1)}) = \text{sgn}(X^{(0)} < X^{(1)})$. By verifying that $\xi_{0,1} = \operatorname{Var}(h_{0,1}) = 1/3 > 0$ and applying Theorem \ref{thm:gen_U}, 
we have
${n_1}^{1/2} T_\tau \xrightarrow{d} N(0, 1/3)$ as $n_1 \to \infty
$ under the null hypothesis and $n_1/n\rightarrow 0$.
We can similarly construct the rescaled Spearman's $T_\rho$, which could be verified to be the same with the Kendall's $\tau$.

3. Rescaled Imbalanced Kendall's \( T_m \). We introduce a new statistic \( T_m \) treating cases and controls unequally. The kernel function is defined as
$
h(X_1^{(0)}, \dots, X_m^{(0)}; X^{(1)}) = \text{sgn}( X^{(1)} > m^{-1} \sum_{k=1}^{m} X_{k}^{(0)}),
$
where information from one case and multiple controls is used. The value of \( \xi_{0,1} > 0 \) is distribution-specific, and \( \xi_{0,1} \neq \xi_{1,0} \). In the special case where \( X^{(0)} \) and \( X^{(1)} \) are drawn from \( N(0,1) \), we verify that
$
\xi_{0,1} = \int_{\mathbb{R}} \{1 - 2\Phi(m^{1/2}x)\}^2 \phi(x) \, dx
$
and
$
\xi_{1,0} = \int_{\mathbb{R}} \{1 - 2\Phi((m^2 + m - 1)^{-1/2}x)\}^2 \phi(x) \, dx,
$
where \( \Phi(x) \) and \( \phi(x) \) are the cumulative distribution function and probability density function of \( N(0,1) \), respectively. 
Thus, in this special case, under the null hypothesis and $n_1/n\rightarrow 0$, we have ${n_1}^{1/2} T_m \xrightarrow{d} N\left(0, \int_{\mathbb{R}} \{1 - 2\Phi(m^{1/2}x)\}^2 \phi(x) \, dx \right)\text{as } n_1 \to \infty$.


4. Rescaled Distance Covariance $T_{\rm dcov}$.
The kernel function for the rescaled distance covariance is defined as
\begin{equation*}
    \begin{aligned}
        h_{\operatorname{dcov}}(X_1^{(0)}, X_2^{(0)}, X_1^{(1)}, X_2^{(1)}) &= \left\| X_1^{(0)} - X_1^{(1)} \right\|_2 + \left\| X_1^{(0)} - X_2^{(1)} \right\|_2 \\
        &+ \left\| X_2^{(0)} - X_1^{(1)} \right\|_2 + \left\| X_2^{(0)} - X_2^{(1)} \right\|_2 \\
        &- 2 \left\| X_1^{(0)} - X_2^{(0)} \right\|_2 - 2 \left\| X_1^{(1)} - X_2^{(1)} \right\|_2.
    \end{aligned}
\end{equation*}
\noindent
The rescaled distance covariance $T_{\rm dcov}$ is defined as
\begin{equation*}
    T_{\rm dcov} = \binom{n_0}{2}^{-1} \binom{n_1}{2}^{-1} \sum_{i=1}^{n_0} \sum_{j > i} \sum_{k=1}^{n_1} \sum_{l > k} h_{\operatorname{dcov}}(X_i^{(0)}, X_j^{(0)}, X_k^{(1)}, X_l^{(1)}).
\end{equation*}
\noindent
It is a second-order RIT with $\xi_{0,1} = \xi_{1,0} = 0$ and $\xi_{0,2} > 0$. Define $D_{\operatorname{dcov}}(v) = \mathrm{E} \| v - X\|_2$ and $\gamma^* = \mathrm{E}\{D_{\operatorname{dcov}}(\tilde{X})\} $, where $\tilde{X}$ is an independent copy of $X$. Applying Theorem \ref{thm:gen_U}, we derive the following conclusion.
	Under the null hypothesis and $n_1/n\rightarrow 0$, (1) for fixed $p$, we have
	$n_1 T_{\rm dcov}  \xrightarrow{d} \eta_{\operatorname{dcov}}$ as $n_1 \to\infty$, where $\eta_{\operatorname{dcov}} = \sum_{j=1}^{\infty} \lambda^*_j (\chi_{1j}^2 - 1)$, $\lambda^*_j$s are weights depending on the distribution of $X^{(1)}$ and the projection kernel function $h_{0,2}(x,y) = D_{\operatorname{dcov}}(x) + D_{\operatorname{dcov}}(y) - \left\|x - y\right\|_2 - \gamma^*$; (2) for diverging $p$ satisfying (\ref{eq:condition_clt}), we have \( n_1 T_{\rm dcov} /\xi_{0,2}^{1/2} \xrightarrow{d} N(0,2) \), where
	$\xi_{0,2} = \mathrm{Var}\left\lbrace  h_{0,2} (X_1 ,X_2) \right\rbrace $,
	\( X_1 \) and \( X_2 \) are independent and identically distributed as \( X \).

5. Rescaled Projection Covariance $T_{\rm IPcov}$.
The kernel function is defined as
\begin{equation*}
    \begin{aligned}
        &h_{\operatorname{IPcov}}(X_1^{(0)}, X_2^{(0)}, X_1^{(1)}, X_2^{(1)}) = A(X_1^{(0)}, X_1^{(1)}) + A(X_1^{(0)}, X_2^{(1)})\\
        &~~~~~~~~+ A(X_2^{(0)}, X_1^{(1)}) + A(X_2^{(0)}, X_2^{(1)}) - 2 A(X_1^{(0)}, X_2^{(0)}) - 2 A(X_1^{(1)}, X_2^{(1)}),
    \end{aligned}
\end{equation*}
where $ A(X_1, X_2) = \mathrm{arccos} \left\{ (c_{\sigma^2} + X_1^{\top} X_2) (c_{\sigma^2} + X_1^{\top} X_1)^{-1/2}(c_{\sigma^2} + X_2^{\top} X_2)^{-1/2} \right\}$. Applying Theorem \ref{thm:gen_U}, we derive the similar conclusion with that for $T_{\rm dcov}$.
    Under the null hypothesis and $n_1/n\rightarrow 0$:
    (1) for fixed $p$, we have
    $
    n_1 T_{\rm IPcov} \xrightarrow{d} \eta_{\operatorname{IPcov}}$ as $n_1 \to \infty,
    $
    where $\eta_{\operatorname{IPcov}} = \sum_{j=1}^{\infty} \lambda^{\dagger}_j (\chi_{1j}^2 - 1)$ with $\lambda^{\dagger}_j$ correspondingly defined;
    (2) for diverging $p$ satisfying (\ref{eq:condition_clt}), we have
    $n_1 T_{\rm IPcov} / \xi_{0,2}^{1/2} \xrightarrow{d} N(0, 2)$, with appropriately defined $\xi_{0,2}$.

In fact, the RIT theory proposed in this work encompasses a broader range of test statistics. In the case of rare events where \(n_0\) and \(n_1\) satisfy certain conditions, as discussed, the convergence rate of the statistics is primarily determined by \(n_1\). For other statistics, it is also possible to define the corresponding kernel functions and derive the theoretical properties of these test statistics within the context of rare events.

\black

\csubsection{Boosted Independence Test based on Subsampling}

Theorem \ref{thm:gen_U} shows that the convergence rate of $T$ depends solely on the number of events $n_1$. This suggests that increasing the sample size for controls $n_0$, contributes the majority of the computational complexity but does not provide much improvement in the convergence rate.
If we can verify that a subsampled statistic exhibits asymptotic properties similar to those of the full sample, computational complexity could be significantly reduced by employing subsampling.

In this subsampling process, we retain all cases (i.e., $Y_i = 1$) and select a subset of controls (i.e., $Y_i = 0$). Let $\delta_i$ be a binary indicator for whether the $i$-th control is included in the sample, where $1 \leqslant i \leqslant n_0$. If the $i$-th control is included, define $\delta_i = 1$; otherwise, $\delta_i = 0$. We assume $\delta_i$ follows an independent Bernoulli distribution, i.e., $\delta_i \sim \operatorname{Bernoulli}(s n_1 / n_0)$, where $s \geqslant 2$ is a predefined integer. Given the kernel function $h(\cdot)$, we define the subsampling statistic, referred to as the boosted independent test statistic based on subsampling (BIT), as
\begin{equation}
    \label{eq:sample_U_form}
    T_{\rm S} = \binom{s n_1}{m_0}^{-1} \binom{n_1}{m_1}^{-1} \sum_c \left(\prod_{j=1}^{m_0} \delta_{i^0_j}\right) h\left(X_{i^0_1}^{(0)}, \cdots, X_{i^0_{m_0}}^{(0)} ; X_{i^1_1}^{(1)}, \cdots, X_{i^1_{m_1}}^{(1)}\right).
\end{equation}
\noindent
Notably, the statistic defined in \eqref{eq:sample_U_form} does not constitute a classical generalized $U$ statistic, since its randomness arises from both $X$ and $\delta$, and the term in the denominator, $s n_1$, does not precisely represent the actual number of sampled controls.
In (\ref{eq:sample_U_form}), incorporating $\delta_i$ into the summation ensures that a term is included only when all observations within the kernel function belong to the subsampled set. This reduces the number of non-zero terms in the summation from $\binom{n_0}{m_0} \binom{n_1}{m_1}$ to $\binom{s n_1}{m_0} \binom{n_1}{m_1}$. The statistic in (\ref{eq:sample_U_form}) is computed using $n_1$ cases and approximately $s n_1$ controls. This leads to the following theorem.

\begin{theorem} (The asymptotic properties of $T_{\rm S}$)
    \label{thm:sample_U}
    Under the null hypothesis and $n_1/n\rightarrow 0$:
    (1) If $\xi_{0,1} > 0$ or $\xi_{1,0} > 0$, and $\mathrm{E}\{ h_{1,0}^4(X_1^{(0)}) \} < \infty$, then the BIT in (\ref{eq:sample_U_form}) satisfies
    \begin{equation*}
        {n_1}^{1/2} T_{\rm S} \xrightarrow{d} N(0, {m_1}^2 \xi_{0,1} + m_0^2 \xi_{1,0} / s).
    \end{equation*}

    (2) If $\xi_{1,0} = \xi_{0,1} = 0$ but $\xi_{0,2} > 0$, $\xi_{2,0} > 0$, or $\xi_{1,1} > 0$, then $n_1 T_{\rm S}$ will converge in distribution to a nondegenerate distribution with zero mean and variance
    $m_1^2 (m_1 - 1)^2 \xi_{0,2} / 2 + m_1^2 m_0^2 \xi_{1,1} / s + m_0^2 (m_0 - 1)^2 \xi_{2,0} / 2 s^2.
    $
\end{theorem}
{\begin{table}[htb]
\begin{center}
	\caption{Comparison of Computational Complexities for Different Test Statistics between RIT and BIT.}
	\label{tab:computation}
	\renewcommand\arraystretch{1.5}
		\begin{tabular}{cc|cc}
			\hline \hline
			  RIT & Computational  & BIT
			&Computational\\
           Notation&Complexity&Notation&Complexity\\
			\hline
			$T_r$& $O\left( n\right)$ & $T_{r,\rm{S}}$& $O\left(n_1s \right)$ \\
			$T_{\tau}$& $O\left( n\log n\right) $& $T_{\tau,\rm {S}}$& $O\left\{ n_1s \log \left( n_1s  \right) \right\} $\\
           $T_{\rm{dcov}}$& $O\left( p n^2 \right)$& $T_{\rm{dcov,S}}$& $O\left( p  n_1^2 s^2 \right)$\\
            $T_{\rm{IPcov}}$& $O\left( p n^2\right)$ & $T_{\rm{IPcov,S}}$& $O\left( p n_1^2 s^2 \right)$ \\
			\hline
		\end{tabular}
\end{center}
\end{table}
}

\noindent
See Appendix C for the detailed proof. From Theorem \ref{thm:sample_U}, we find that the convergence rates of RIT and BIT are the same. The properties of the BIT defined in formula (\ref{eq:sample_U_form}) depend only on $n_1$ and $s$.
The larger the sampling ratio, the smaller the variance of the asymptotic distribution. As $s$ tends to infinity, the conclusion of Theorem \ref{thm:sample_U} aligns with that of Theorem \ref{thm:gen_U}. However, as $s$ increases, the computational complexity grows proportionally to $s$, while the reduction in variance becomes less significant. Therefore, we should consider the balance between reducing variance and increasing computational complexity. Table \ref{tab:computation} compares the computational complexities of the rescaled test statistics and their sampled counterparts.



While the subsampling method reduces computational complexity, unlike Theorem \ref{thm:gen_U}, Theorem \ref{thm:sample_U} does not provide the asymptotic distribution of BIT \( T_{\rm S} \) under \( \xi_{1,0} = \xi_{0,1} = 0 \). Thus, computationally intensive simulations are still required to approximate the null distribution.
In a specific high-dimensional case, the asymptotic distribution could be explicitly expressed in the following theorem.

\begin{theorem} (The asymptotic properties of high-dimensional BIT)
    \label{thm:sample_div_p}
    Under the null hypothesis and $n_1/n\rightarrow 0$, if $\xi_{1,0} = \xi_{0,1} = 0$ but $\xi_{0,2} > 0$, and $h_{0,2}(\cdot,\cdot)$ depends on $n_1$ satisfying \eqref{eq:condition_clt}, as $n_1 \to \infty$. With $s \to \infty$, the BIT $n_1 T_{\rm S} / \xi_{0,2}^{1/2}$ is asymptotically normally distributed with zero mean and variance $m_1^2 (m_1 -1)^2 / 2$.
\end{theorem}

\noindent
The proof of this theorem is given in Appendix C. The conclusions of Theorem \ref{thm:sample_div_p} are similar to those of Theorem \ref{thm:T_div_p}. Given \( n_1 / n_0 \to 0 \), we can select \( s \) such that \( s n_1 = o(n_0) \), ensuring reduced computational complexity while keeping all rare events \( Y = 1 \) in the sampled data. However, in high-dimensional settings, \( s \) must diverge rather than remain fixed, although it may diverge at a much slower rate. Additionally, we must verify the conditions in \eqref{eq:condition_clt} for \( h_{0,2} \) to apply the asymptotic distribution in Theorem \ref{thm:sample_div_p}. Using both simulated and real data, we demonstrate in the numerical studies that when $s \geqslant 5$, the performance closely approximates that of the full sample.

\textbf{Remark: (Discussion on selection of $s$)} We provide the verification about BIT power in Appendix C.6. For the first-order BIT, we consider the choice of \( s \) from three perspectives:
(1) (Asymptotic Variance): The difference between the asymptotic variances of \( T \) and \( T_{\rm S} \) is \( O(n_1^{-1} s^{-1}) \). To control this difference to be no more than \( \epsilon \), the suggested setting is \( s \geqslant (n_1 \epsilon)^{-1} \). A larger \( \epsilon \) reduces complexity, while a smaller \( \epsilon \) brings the power of \( T_{\rm S} \) closer to that of \( T \).
(2) (Lower Bound of Power): If the power of \( T_{\rm S} \) needs to be no less than \( \beta \), from Appendix C.6, we calculate that \( s \geqslant m_0^2 \xi_{1,0} / \left[ n_1 \mu_0^2 / \{\Phi^{-1}(1-\alpha/2) - \Phi^{-1}(1- \beta)\}^2 - m_1^2 \xi_{0,1} \right] \), where \( \mathrm{E}h(\cdot) = \mu_0 > 0 \) and \( s \) is approximately inversely proportional to \( n_1 \).
(3) (Power Difference): If the power difference between \( T \) and \( T_{\rm S} \) should be no larger than \( \epsilon \), then it could be calculated that \( s \geqslant n_1^{1/2} m_0^2 \mu_0 \xi_{1,0} \phi\left\{ \Phi^{-1}(1-\alpha/2) - \mu_0 (n_1/m_1^2 \xi_{0,1})^{1/2} \right\} / 2 \epsilon m_1^3 \xi_{0,1}^{3/2} \). When \( n_1 \) is sufficiently large, \( \phi\left\{ \Phi^{-1}(1-\alpha/2) - \mu_0 (n_1/m_1^2 \xi_{0,1})^{1/2} \right\} \approx \exp\left( - \mu_0^2 n_1 / 2 m_1^2 \xi_{0,1} \right) \). This implies that \( s \) is approximately proportional to \( n_1^{1/2} \exp(-n_1) \). Thus, a smaller \( n_1 \) requires a larger \( s \), and vice versa. For the second-order BIT, it could be verified that in the case of high-dimensionality, the statistic \( T_{\rm S} \) can achieve the same power as \( T_{\infty} \), see Appendix C.6 for details. Similarly, we observe that the power of \(T_{\infty}\) is only related to $n_1$ as well. The analysis of the power above explains the phenomenon in Figure \ref{fig:rare_event}. Additionally, we demonstrated in the simulation section that the loss in empirical power at different \(s\) is relatively small.

{\csubsection{Local Power Analysis}

In this subsection, we compare the local power of RIT and BIT, where local power is crucial for independence testing \citep{shi2022power}.
It examines the capability of the test statistic to reliably detect a non-independence relationship, where the signal violating the null hypothesis diminishes as the sample size increases. We consider the classical and well-used mixture  alternatives  \citep{farlie1960performance,farlie1961asymptotic,dhar2016study}. Specifically, the conditional distribution function of $X \mid Y = 0$ is $F(x)$, while the conditional distribution function of $X \mid Y = 1$ is $(1-\Delta)F(x) + \Delta G(x)$, where $0 < \Delta \leqslant 1$, and $G(x)$ is a fixed distribution function that is absolutely continuous and satisfies $G(x) \neq F(x)$.
This implies that a new sample $X_i$ with $Y_i = 1$ is drawn from $F(x)$ with probability $1-\Delta$, and from $G(x)$ with probability $\Delta$. Since the kernel function $h(X_1^{(0)}, \cdots, X_{m_0}^{(0)}; X_1^{(1)}, \cdots, X_{m_1}^{(1)})$ has zero mean only when $X_i^{(1)}$ and $X_j^{(0)}$ are from the same distribution, $h(\cdot)$ no longer has zero mean if there exists $k \in \{1, \cdots, m_1\}$ such that $X_k^{(1)}$ is drawn from $G(x)$. Denote $\mathrm{E} h(X_1^{(0)}, \cdots, X_{m_0}^{(0)}; X_1^{(1)}, \cdots, X_{m_1}^{(1)} \mid X_i^{(1)}\sim G(x) \text{ for } 1\leqslant i \leqslant q \text{ and } X_j^{(1)} \sim F(x) \text{ for } (q+1) \leqslant j \leqslant m_{1} ) \overset{\triangle}{=} \mu_{G,q}$. We assume that $\mu_{G,1}\neq 0$ and $\mu_{G,q}< \infty$ exists for $2\leqslant q \leqslant m_1-1$.

First, for the first order  RIT $T$ and BIT $T_{\rm S}$, we examine the asymptotic power under the following sequence of alternatives: $\mbox{H}_{1,n_1}(\Delta_{0})$: $\Delta = \Delta_{n_1}$ with $\Delta_{n_1} = n_1^{-1/2} \Delta_0$, where $\Delta_0 > 0$ is a constant.
Then applying the same techniques as in Theorem \ref{thm:gen_U}, we obtain \(n_1^{1/2} T  \xrightarrow{d} N(m_1 \Delta_0 \mu_{G,1}, m_1^2 \xi_{0,1})\) and \(n_1^{1/2} T_{\rm S} \xrightarrow{d} N(m_1 \Delta_0 \mu_{G,1}, m_1^2 \xi_{0,1} + m_0^2 \xi_{1,0} / s)\). We now state the following corollary.
 \begin{corollary} (The power of $T$ and $T_{\rm S}$)
 	\label{thm: power_local}
 	Assume the local alternative hypothesis $\mbox{H}_{1,n_1}(\Delta_{0})$ is true with $n_1/n_0\rightarrow 0$. Then for any $0 < \beta < 1-\alpha $, (1) there exists $C(\beta,\alpha,\mu_{G,1},\xi_{0,1}) > 0$ such that as long as $\Delta_0 > C(\beta,\alpha,\mu_{G,1},\xi_{0,1})$, we have,
 	$$ \lim_{n_1 \to \infty} \mathrm{P}\left\lbrace \phi_T(T) = 1 \mid \mbox{H}_{1,n_1}(\Delta_{0})  \right\rbrace \geqslant 1-\beta , $$
 	where $\phi_T(T)$ is the test function for $T$;
 	
 	(2) there exists $C(\beta,\alpha,\mu_{G,1},\xi_{0,1}+ m_0^2 \xi_{1,0}/s m_1^2) > 0$ such that as long as $\Delta_0 > C(\beta,\alpha,\mu_{G,1},\xi_{0,1}+ m_0^2 \xi_{1,0}/s m_1^2)$,
 	$$ \lim_{n_1 \to \infty} \mathrm{P}\left\lbrace \phi_{T_{\rm S}}(T_{\rm S}) = 1 \mid \mbox{H}_{1,n_1}(\Delta_{0})  \right\rbrace \geqslant 1-\beta.$$
 \end{corollary}

\noindent
The detailed verification of Corollary \ref{thm: power_local} is provided in Appendix C.7. We conclude that even with a diminishing signal as $n_1$ increases, satisfactory power can still be achieved. On the other hand, if we consider the classical sequence of alternatives with $\Delta_{n} = n^{-1/2} \Delta_0$, we find that for any fixed constant $\Delta_0 > 0$, $\lim_{n \to \infty} \mathrm{P}\left\lbrace \phi_{T}(T) = 1 \mid \mbox{H}_{1,n}(\Delta_{0})  \right\rbrace = \alpha$. In rare event scenario, the effective sample size is primarily determined by $n_1$. Consequently, we can detect signals that decay at the rate of  $n_1^{-1/2}$. However, when the signal decays at the rate of $n^{-1/2}$, the growth rate of the effective sample size cannot keep pace with the rate of signal decay, resulting in the failure of the test statistic.

For the second order statistics, we focus solely on the case of high-dimensionality, as it is not easy to obtain the explicit asymptotic distribution for fixed dimension. We examine the asymptotic power of $T_{\infty}$ under the following sequence of alternatives: $\mbox{H}^{\infty}_{1,n_1}(\Delta_{0})$: $\Delta = \Delta_{n_1}$ with $\Delta_{n_1} = n_1^{-1} \Delta_0$.
To ensure that $F(x)$ and $G(x)$ remain distinguishable when $p\rightarrow \infty$,
we assume $\liminf_{n_1\to \infty,p\to \infty} |\mu_{G,1}| \geqslant c_{G} >0$.
 By applying the same techniques as in Corollary \ref{thm: power_local}, we can obtain a similar result: assuming the local alternative hypothesis $\mbox{H}^{\infty}_{1,n_1}$ holds for \( p \to \infty \), and \( h_{0,2}(\cdot, \cdot) \) satisfies condition \eqref{eq:condition_clt}, we have $\lim_{n_1 \to \infty} \mathrm{P}\lbrace \phi_{T_{\infty}}(T_{\infty}) = 1 \mid \mbox{H}_{1,n_1}(\Delta_{0}) \rbrace \geqslant 1-\beta$.
 With $s \rightarrow \infty$, it can be verified that $T_{\rm S}$ will achieve the same local power as $T_{\infty}$.
}

\csubsection{Boosted Independence Test with Multiple Rare Event Classes}

In the aforementioned analysis, we only considered $Y$ as a binary variable. However, in real applications, there are many situations where $Y$ is multi-class. Among the classes of $Y$, one class typically serves as the majority, while the other classes represent rare events.
For example, in the analysis of drug treatment effects, the vast majority of people do not experience adverse outcomes. Those who do experience adverse outcomes constitute a small minority and can be further categorized into different types. Similarly, in the analysis of anomalies in financial transactions, most transactions are normal, forming the majority class. In contrast, a small number of anomalous transactions, which can be classified into various types, represent the rare events.

Here, we define $Y_i \in \{0,1,\cdots,K\}$ as the class label for the $i$-th observation. In this way, all $X$ observations can be classified into $K+1$ distinct sets, defined as $\mathcal{D}^k=\{X_1^{(k)}, X_2^{(k)}, \cdots, X_{n_k}^{(k)}\}$, where $n_k$ is the number of observations in each set for $0 \leqslant k \leqslant K$. Thus, we consider the test statistics as,
\begin{equation*}
 T_{K}= \prod_{j=0}^K\binom{n_j}{m_j}^{-1}  \sum_c h\left(X_{i^0_{1}}^{(0)}, \cdots,X_{i^0_{m_0}}^{(0)} ; \cdots ;X_{i^K_{1}}^{(K)}, \cdots,X_{i^K_{m_K}}^{(K)}\right) ,
\end{equation*}
where $h(\cdot)$ is a kernel function assumed to be symmetric within each of its $K+1$ blocks. The set $\{i^k_{1}, \ldots, i^k_{m_k}\}$ denotes $m_k$ distinct elements from the set $\{1,2, \ldots, n_k\}$ for $0 \leqslant k \leqslant K$, and $\sum_c$ denotes summation over all possible combinations. This statistic is referred to as Multi-RIT. To further establish the theoretical properties of $T_{K}$, we define the projection of the kernel function as
\begin{equation*}
	\begin{aligned}
	&h_{\nu_0,\nu_1,\cdots,\nu_K} \left( X_{i^0_1}^{(0)} , \cdots,X_{i^0_{\nu_0}}^{(0)};\cdots,X_{i^K_{1}}^{(K)} , \cdots,X_{i^K_{\nu_K}}^{(K)}\right)  \\
	= & \mathrm{E}\left\lbrace  h\left(X_{i_1^0}^{(0)}, \cdots,X_{i_{m_0}^0}^{(0)} ;\cdots;X_{i_1^K}^{(K)}, \cdots,X_{i^K_{m_K}}^{(K)}\right) \mid X_{i^0_1}^{(0)} , \cdots,X_{i^0_{\nu_0}}^{(0)};\cdots; X_{i^K_{1}}^{(K)} , \cdots,X_{i^K_{\nu_K}}^{(K)} \right\rbrace,
\end{aligned}
\end{equation*}
for $0 \leqslant \nu_k < m_k$ with $0 \leqslant k \leqslant K$. We then investigate the theoretical properties of this statistic in the scenario of multi-class rare events. For simplicity, we assume $n_k = o(n_0)$ for $1 \leqslant k \leqslant K$. Further, we define the following notations for multi-class rare events, for $1 \leqslant k \leqslant K$ and $1 \leqslant k_1 < k_2 \leqslant K$,
\begin{equation*}
	\begin{aligned}
	\phi_{1,k} \left( X_1^{(k)}\right) &= \underbrace{h_{0,\cdots,0,1,0,\cdots,0}}_{\text{only } \nu_k \text{ is 1}} \left( X_1^{(k)}\right), \end{aligned}
\end{equation*}
\begin{equation*}
	\begin{aligned}
	\phi_{2,k} \left( X_1^{(k)} ; X_2^{(k)}\right) &= \underbrace{h_{0,\cdots,0,2,0,\cdots,0}}_{\text{only } \nu_k \text{ is 2}} \left( X_1^{(k)} ; X_2^{(k)}\right),\\
	\phi_{2,k_1,k_2} \left( X_1^{(k_1)} ; X_1^{(k_2)}\right) &= \underbrace{h_{0,\cdots,0,1,0,\cdots,0,1,0,\cdots,0}}_{\text{only } \nu_{k_1}\text{ and } \nu_{k_2} \text{ are 1}} \left( X_1^{(k_1)} ; X_1^{(k_2)}\right).
\end{aligned}
\end{equation*}
Then define $\zeta_{1,k} = \operatorname{Var}\{ \phi_{1,k}( X_1^{(k)}) \} $, $\zeta_{2,k_1,k_2} = \operatorname{Var}\{ \phi_{2,k_1,k_2} ( X_1^{(k_1)} ; X_1^{(k_2)}) \}$, and
	$\zeta_{2,k} = \operatorname{Var}\{ \phi_{2,k} ( X_1^{(k)} ; X_2^{(k)}) \}$, respectively. We have the following theorem.
\begin{theorem} (The asymptotic properties of $T_K$)
    \label{thm:multiple_RIT}
    Under the assumption of the null hypothesis and the condition of rare events \( n_k = o(n_0) \), assume that \( {n_k}/{n_1} \to r_k \) for \( 2 \leqslant k \leqslant K \), where \( r_k \)s are non-zero constants.
    (1) If there exists $1 \leqslant k \leqslant K$ such that $\zeta_{1,k} > 0$, then the Multi-RIT satisfies
    $$ {n_1}^{1/2} T_{K} \xrightarrow{d} N\left( 0, \sum_{k=1}^{K} {m_k}^2 \zeta_{1,k}/r_k\right) .$$

    \noindent
    (2) When $\zeta_{1,k} = 0$ for all $0 \leqslant k \leqslant K$, and there exists some $\zeta_{2,k_1,k_2} > 0$ or $\zeta_{2,k} > 0$ for $1 \leqslant k_1 < k_2 \leqslant K$ and $1 \leqslant k \leqslant K$, $n_1 T_{K}$ will converge in distribution to a nondegenerate distribution with zero mean and variance
    $ \sum_{k= 1}^K m_{k}^2 (m_{k} - 1)^2 \zeta_{2,k}/2 r_{k}^2 + \sum_{1 \leqslant k_1 < k_2 \leqslant K} m_{k_1}^2 m_{k_2}^2 \zeta_{2,k_1,k_2}/r_{k_1} r_{k_2}. $
\end{theorem}

\noindent
The proof of Theorem \ref{thm:multiple_RIT} is given in Appendix C.  We conclude that when sample sizes across rare classes are comparable, the asymptotic distribution of \( T_{K} \) depends on \( n_1 \), regardless of the sample size in the majority class. For binary classification (\( K = 1 \)), Theorem \ref{thm:multiple_RIT} simplifies to Theorem \ref{thm:gen_U}.
In practical applications, subsampling techniques can also be used to obtain sample statistics with the same convergence rate in cases with multiple rare classes. We refer to this method as Multi-BIT. Given the complexity of rare events with varying class sizes, for simplicity, we consider the case where all sample sizes are comparable, with $n_1 / n_k \to 1$ for any $2 \leqslant k \leqslant K$. In this scenario, the Multi-BIT can be defined as follows,
\begin{equation*}
	T_{\rm{S},K}= \binom{s n_1}{m_0}^{-1} \prod_{k=1}^K\binom{n_k}{m_k}^{-1} \sum_c { \left(\prod_{j=1}^{m_0}\delta_{i^0_j}\right)}h\left(X_{i^0_1}^{(0)}, \cdots,X_{i^0_{m_0}}^{(0)} ;\cdots;X_{i^K_1}^{(1)}, \cdots,X_{i^K_{m_K}}^{(1)}\right).
\end{equation*}

\noindent
In this case, we have the following corollary.
\begin{corollary} (The asymptotic properties of $T_{\rm{S},K}$)
    \label{thm:sample_multiple_RIT}
    Under the null hypothesis and assuming the same conditions as in Theorem \ref{thm:multiple_RIT},
    (1) If there exists $0 \leqslant k \leqslant K$ such that $\zeta_{1,k} > 0$, then $T_{\rm{S},K}$ satisfies ${n_1}^{1/2} T_{\rm{S},K} \xrightarrow{d} N( 0, {m_0}^2 \zeta_{1,0}/s + \sum_{k=1}^{K} {m_k}^2 \zeta_{1,k} )$.
    \noindent
    (2) When $\zeta_{1,k} = 0$ for all $0 \leqslant k \leqslant K$, and there exists some $\zeta_{2,k_1,k_2} > 0$ or $\zeta_{2,k} > 0$ for $0 \leqslant k_1 < k_2 \leqslant K$ and $0 \leqslant k \leqslant K$, $n_1 T_{\rm{S},K}$ will converge in distribution to a nondegenerate distribution with zero mean and variance
    $
    \sum_{k=1}^K m_{k}^2 (m_k - 1)^2 \zeta_{2,k}/2 + \sum_{1 \leqslant k_1 < k_2 \leqslant K} m_{k_1}^2 m_{k_2}^2 \zeta_{2,k_1,k_2} + m_0^2(m_0-1)^2 \zeta_{2,0}/ 2 s^2 + \sum_{k=1}^K m_0^2 m_k^2 \zeta_{2,0,k}/s .
    $
\end{corollary}

\noindent
The proof of this corollary is similar to those of Theorem \ref{thm:sample_U} and Theorem \ref{thm:multiple_RIT}, and thus omitted. It is noteworthy that the last two terms, $m_0^2(m_0-1)^2 \zeta_{2,0}/ 2s^2$ and $\sum_{k=1}^K m_0^2 m_k^2 \zeta_{2,0,k}/s$, are induced by subsampling. Either $\zeta_{2,0} \neq 0$ or $\zeta_{2,0,k} \neq 0$ will result in a non-zero increase in the variance of the statistic. However, if $s$ is sufficiently large, this increase may be negligible compared to the leading terms.

 Similarly, when there is only one rarest class (without loss of generality, the first class $k=1$) containing the smallest sample size (i.e., \(n_1/n_k \rightarrow 0\) for \(k=0, 2, 3, \ldots, K\)), the convergence rate is primarily determined by the sample size of this rarest class $n_1$. The sample-based version can be similarly analyzed and is omitted here for simplicity. The proof of this corollary is similar to that of Theorem \ref{thm:multiple_RIT} and thus omitted.

\begin{corollary} (The asymptotic properties of $T_K$ with one rarest class)
    \label{thm:special case multiple_RIT}
    Under the null hypothesis and the condition of rare events \( n_1 = o(n_k) \) for \( k=0,2,3,\cdots, K \), we have:
    (1) If $\zeta_{1,1} > 0$, then the Multi-RIT satisfies
    $\sqrt{n_1}  T_{K} \xrightarrow{d} N\left( 0, {m_1}^2 \zeta_{1,1}\right)$;
    \noindent
    (2) When $\zeta_{1,k} = 0$ for all $0 \leqslant k \leqslant K$, and there exists some $\zeta_{2,1,k} > 0$ for $k = 0,2,3,\cdots,K$ or $\zeta_{2,1} > 0$, $n_1 T_{K}$ will converge in distribution to a nondegenerate distribution with zero mean and variance
    $m_{k}^2 (m_{1} - 1)^2 \zeta_{2,1}/2 + \sum_{k \neq 1} m_{1}^2 m_{k}^2 \zeta_{2,1,k}.$
\end{corollary}

\csection{EMPIRICAL ANALYSIS}
\csubsection{Simulation Settings}

In this section, we conduct simulations to compare the performance of different RITs in rare event scenarios with the corresponding classical test statistics. We consider four rescaled independence test statistics, constructed based on Pearson correlation ($T_{r}$), Kendall's $\tau$ rank correlation (Kendall $T_{\tau}$), distance correlation (DC $T_{\rm dcov}$), and improved projection correlation (IPC $T_{\rm IPcov}$).
Three types of test statistics are considered: (1) classical test statistics (denoted as Original); (2) rescaled test statistics (RIT); and (3) boosted test statistics based on subsampling (BIT). We design different studies to generate observations for first-order RITs ($T_{r}$ and $T_{\tau}$) and second-order RITs ($T_{\rm dcov}$ and $T_{\rm IPcov}$), as outlined below.

1. First-Order RIT. We fix $p=1$ and consider two examples.
Example 1: For cases with $Y_i=1$, the covariates are generated as $X_i \overset{\text{iid}}{\sim} N(0,1)$; for controls with $Y_i=0$, they are generated as $X_i \overset{\text{iid}}{\sim} N(0.3,1)$.
Example 2: Assume $X_i \overset{\text{iid}}{\sim} N(0,1)$ and $Y_i \overset{\text{iid}}{\sim} \operatorname{Bern}(p_i)$, where $p_i = \operatorname{logit}(\beta_0 + \beta_1 X_i)$, with $\beta_1 = 0.3$ and $\beta_0 = \log(n_1/n_0)$, ensuring $\mathrm{E} Y_i \approx n_1/n$. Randomly change some class labels to ensure the number of cases equals \(n_1\).

2. Second-Order RIT. We focus on large-scale datasets and consider the high-dimensional case with $p=50$ in the following two examples.
Example 1: For \(Y_i = 1\), generate \(X_i \in \mathbb{R}^p\) from a multivariate normal distribution with mean \(\mathbf{0}_p \in \mathbb{R}^p\) (all entries zero) and covariance matrix \(\Sigma = (0.5^{|j_1 - j_2|})_{p \times p}\), where $1 \leq j_1, j_2 \leq p$. For \(Y_i = 0\), generate \(X_i\) from a multivariate normal distribution with mean \(0.4 \times (\mathbf{1}_{p_1}^\top, \mathbf{0}_{p-p_1}^\top)^\top \in \mathbb{R}^p\) and the same covariance matrix \(\Sigma\), where $p_1 = 10$ and $\mathbf{1}_{p_1} \in \mathbb{R}^{p_1}$ is a vector of ones.
Example 2: Let \(X_i \in \mathbb{R}^p\) follow a multivariate normal distribution with mean \(\mathbf{0}_p\) and covariance matrix \(\Sigma\). Assume $Y_i \overset{\text{iid}}{\sim} \operatorname{Bern}(p_i)$, where $p_i = \operatorname{logit}(\beta_0 + \sum_{i=1}^{p} \beta_i X_i)$. Here, $\beta_i = 0.15$ for $1 \leq i \leq 10$ and $\beta_i = 0$ otherwise, with \(\beta_0 = \log(n_1/n_0)\).

We conduct two studies to demonstrate the breakdown of classical tests in rare event scenarios and evaluate the performance of the rescaled independence test.

Study 1. In this study, we illustrate the empirical size and power of the classical test statistics in rare events. We consider two scenarios to measure the effect of varying $n$ and $n_1$, respectively. For Pearson and Kendall's $\tau$ correlations, we first fix \(n_1 = 50\) and vary the total sample size \(n = 100, 1,000, 5,000, 10,000, 20,000\). Second, we fix \(n = 20,000\) and vary the number of positive cases \(n_1 = 50, 100, 500, 1,000, 2,000\). For the distance correlation and improved projection correlation, we fix \(n_1 = 50\) and vary the total sample size \(n = 200, 500, 1,000, 2,000, 5,000\). Then, we fix \(n = 5,000\) and vary the number of positive cases \(n_1 = 50, 100, 200, 500, 1,000\).

Study 2. In this study, we compare the empirical size and power of the RIT and BIT.
For the First-Order RIT, we fix $n = 100,000$ and consider different numbers of cases as $n_1 = 50, 100, 200, 500$ with $n_1 s = 500, 1,000, 2,000, 4,000$.
For the Second-Order RIT, we fix $n = 10,000$ and consider different numbers of cases as $n_1 = 30, 50, 70, 100$ with $n_1 s = 1,000, 1,500, 2,000, 2,500$.

\csubsection{Evaluation and Simulation Result}

For reliable evaluation of different test statistics across methods, we repeat each setting $M = 1000$ times. The performance is measured using the empirical rejection probability (ERP), calculated as $\operatorname{ERP} = M^{-1} \sum_m I(T_{\rm{stat}}^{(m)} \in \mathcal{C}_{\rm{stat}})$, where $T_{\rm{stat}}^{(m)}$ denotes the test statistic in the $m$-th replication, and $\mathcal{C}_{\rm{stat}}$ is the corresponding rejection region. Depending on whether $X$ and $Y$ are independent, the ERP reflects either empirical size or power. All tests are conducted at a significance level of $\alpha = 0.05$. The simulation results are summarized in Tables \ref{tab:Classical_Statistics}--\ref{tab:computation_pearson_and_kendall} and Table \ref{tab:computation_dcor_and_ipcor} in Appendix D.

Study 1. We first evaluate the performance of classical test statistics, as reported in Table \ref{tab:Classical_Statistics}. The results are consistent across different tests. First, the empirical sizes of Pearson correlation and Kendall's $\tau$ remain stable around 0.05, while those of distance correlation and improved projection correlation tend to be higher. Second, when $n_1$ is fixed, increasing the overall sample size does not improve power: Pearson and Kendall's $\tau$ maintain power between 0.5 and 0.6, while distance and projection correlations stabilize between 0.8 and 0.9. Third, increasing $n_1$ improves power, and when $n_1$ reaches 500, the power of all test statistics reaches 1. These results highlight that, in rare event scenarios, test effectiveness is determined by $n_1$ rather than the total sample size.

Study 2. We now examine the performance of the rescaled independence test statistics and the boosted sampled versions. Table \ref{tab:computation_pearson_and_kendall} presents the empirical size and power of rescaled Pearson \(T_r\), rescaled Kendall's \(T_\tau\), and their boosted counterparts based on the subsample across different sample sizes \(n_1 s\).
First, empirical sizes are all around 0.05, consistent with Theorems 1 and 2. Second, the power of $T_r$ and $T_{\tau}$ increases to 1 as $n_1$ grows, confirming the effectiveness of these statistics. Third, when \(n_1 s = 2,000\), the empirical power of $T_{r,\rm{S}}$ and $T_{\tau,\rm{S}}$ closely matches that of $T_r$ and $T_{\tau}$, while using only 2\% of the full sample size, significantly reducing computational costs. Similar findings apply to the second-order RITs in Appendix D Table \ref{tab:computation_dcor_and_ipcor}.

{		\begin{table}[htb]
			\begin{center}
				\caption{Comparison of Empirical Size and Power of Classical Statistics.}
				\label{tab:Classical_Statistics}
				\renewcommand\arraystretch{1.5} 
				\begin{tabular}{l|cccc|cccc}
					\hline \hline
					&\multicolumn{4}{c|}{Empirical Size}&\multicolumn{4}{c}{Empirical Power}\\
					\hline
					&Eg.1&Eg.2&Eg.1&\multicolumn{1}{c|}{Eg.2}&Eg.1&Eg.2&Eg.1&Eg.2\\
					\hline
					&\multicolumn{2}{c}{Pearson $\hat{r}^2$}&\multicolumn{2}{c|}{Kendall $\hat\tau$}&\multicolumn{2}{c}{Pearson $\hat{r}^2$}&\multicolumn{2}{c}{Kendall $\hat\tau$}\\
					\hline
					$n$&\multicolumn{8}{c}{Case 1: Fixed $n_1 = 50$}\\
					100   & 0.045 & 0.042 & 0.064 & 0.050 & 0.325 & 0.280 & 0.308 & 0.275 \\
					1000  & 0.044 & 0.056 & 0.061 & 0.049 & 0.543 & 0.498 & 0.538 & 0.495 \\
					5000  & 0.060 & 0.047 & 0.048 & 0.048 & 0.555 & 0.532 & 0.545 & 0.514 \\
					10000 & 0.049 & 0.038 & 0.047 & 0.039 & 0.566 & 0.532 & 0.552 & 0.521 \\
					20000 & 0.051 & 0.057 & 0.048 & 0.049 & 0.564 & 0.533 & 0.552 & 0.523 \\
					\hline
					$n_1$&\multicolumn{8}{c}{Case 2: Fixed $n=20,000$}\\
					50   & 0.054 & 0.043 & 0.054 & 0.045 & 0.562 & 0.534 & 0.549 & 0.523 \\
					100  & 0.054 & 0.053 & 0.047 & 0.054 & 0.845 & 0.834 & 0.840 & 0.825 \\
					500  & 0.049 & 0.043 & 0.063 & 0.051 & 1.000 & 1.000 & 1.000 & 1.000\\
					1000 & 0.043 & 0.047 & 0.051 & 0.044 & 1.000 & 1.000 & 1.000 & 1.000 \\
					2000 & 0.047 & 0.048 & 0.051 & 0.054 & 1.000 & 1.000 & 1.000 & 1.000 \\
					\hline
					&\multicolumn{2}{c}{DC $\widehat{\rm{dcorr}}$}&\multicolumn{2}{c|}{IPC $\widehat{\rm{IPcorr}}$}&\multicolumn{2}{c}{DC $\widehat{\rm{dcorr}}$}&\multicolumn{2}{c}{IPC $\widehat{\rm{IPcorr}}$}\\
					\hline
					$n$&\multicolumn{8}{c}{Case 1: Fixed $n_1 = 50$}\\
					200  & 0.071 & 0.058 & 0.067    & 0.077    & 0.684 & 0.641 & 0.724    & 0.623    \\
					500  & 0.065 & 0.068 & 0.058    & 0.050   & 0.774 & 0.732 & 0.786    & 0.752    \\
					1000 & 0.073 & 0.070 & 0.056    & 0.058    & 0.864 & 0.824 & 0.808    & 0.780    \\
					2000 & 0.075 & 0.062 & 0.058    & 0.044    & 0.872 & 0.830 & 0.826    & 0.798    \\
					5000 & 0.062 & 0.067 & 0.056    & 0.052    & 0.876 & 0.836 & 0.832    & 0.808    \\
					\hline
					$n_1$&\multicolumn{8}{c}{Case 2: Fixed $n=5,000$}\\
					50   & 0.072 & 0.074 & 0.068    & 0.058 & 0.884 & 0.836 & 0.832    & 0.808 \\
					100  & 0.080 & 0.066 & 0.054    & 0.072 & 0.996 & 0.984 & 0.978    & 0.964 \\
					200  & 0.074 & 0.044 & 0.070    & 0.058 & 1.000 & 1.000 & 1.000    & 1.000 \\
					500  & 0.082 & 0.054 & 0.074    & 0.038 & 1.000 & 1.000 & 1.000    & 1.000 \\
					1000 & 0.060 & 0.048 & 0.052    & 0.054 & 1.000 & 1.000 & 1.000    & 1.000 \\
					\hline
				\end{tabular}
			\end{center}
		\end{table}
}

{
		\begin{table}[htb]
		\begin{center}
			\caption{Comparison of Empirical Size and Power of Rescaled Pearson $T_r$ and  Rescaled Kendall's $T_\tau$ and the corresponding BITs with Different $n_1 s$.}
			\label{tab:computation_pearson_and_kendall}
			\renewcommand\arraystretch{1.5} 
				\begin{tabular}{l|ccccc|ccccc}
					\hline \hline
					&\multicolumn{5}{c|}{Empirical Size}&\multicolumn{5}{c}{Empirical Power}\\
					\hline
					&\multicolumn{1}{c|}{}&\multicolumn{4}{c|}{$n_1  s$}	&\multicolumn{1}{c|}{}&\multicolumn{4}{c}{$n_1  s$}\\
					\hline
					&\multicolumn{1}{c|}{RIT}&500&1000&2000&4000&\multicolumn{1}{c|}{RIT}&500&1000&2000&4000\\
					\hline
					$n_1$&\multicolumn{10}{c}{Rescaled Pearson $T_r$ Eg.1}\\
					50 & 0.055 & 0.046 & 0.052 & 0.046 & 0.047 & 0.569 & 0.500 & 0.533 & 0.541 & 0.556 \\
					100 & 0.049 & 0.040 & 0.053 & 0.044 & 0.056 & 0.845 & 0.773 & 0.817 & 0.837 & 0.843 \\
					200 & 0.056 & 0.051 & 0.067 & 0.050 & 0.051 & 0.987 & 0.955 & 0.964 & 0.981 & 0.986 \\
					500 & 0.055 & 0.048 & 0.049 & 0.057 & 0.051 & 1.000 & 0.997 & 1.000 & 1.000 & 1.000 \\
					\hline
					$n_1$&\multicolumn{10}{c}{Rescaled Kendall's $T_\tau$ Eg.1}\\
					50 & 0.048 & 0.048 & 0.052 & 0.044 & 0.047 & 0.532 & 0.504 & 0.511 & 0.519 & 0.532 \\
					100 & 0.051 & 0.047 & 0.047 & 0.053 & 0.053 & 0.831 & 0.770 & 0.791 & 0.815 & 0.814 \\
					200 & 0.051 & 0.052 & 0.049 & 0.053 & 0.048 & 0.987 & 0.929 & 0.967 & 0.980 & 0.983 \\
					500 & 0.050 & 0.049 & 0.055 & 0.053 & 0.046 & 1.000 & 0.996 & 0.999 & 1.000 & 1.000 \\
					\hline
					$n_1$&\multicolumn{10}{c}{Rescaled Pearson $T_r$ Eg.2}\\
					50 & 0.049 & 0.041 & 0.056 & 0.053 & 0.062 & 0.535 & 0.475 & 0.479 & 0.523 & 0.531 \\
					100 & 0.056 & 0.047 & 0.040 & 0.046 & 0.054 & 0.834 & 0.775 & 0.782 & 0.812 & 0.832 \\
					200 & 0.052 & 0.047 & 0.043 & 0.044 & 0.058 & 0.991 & 0.939 & 0.975 & 0.980 & 0.985 \\
					500 & 0.055 & 0.062 & 0.047 & 0.056 & 0.051 & 1.000 & 0.999 & 1.000 & 1.000 & 1.000 \\
					\hline
					$n_1$&\multicolumn{10}{c}{Rescaled Kendall's $T_\tau$ Eg.2}\\
					50 & 0.051 & 0.058 & 0.050 & 0.052 & 0.068 & 0.516 & 0.480 & 0.489 & 0.513 & 0.515 \\
					100 & 0.050 & 0.044 & 0.048 & 0.054 & 0.034 & 0.811 & 0.740 & 0.775 & 0.799 & 0.809 \\
					200 & 0.056 & 0.042 & 0.058 & 0.042 & 0.054 & 0.984 & 0.928 & 0.961 & 0.972 & 0.980 \\
					500 & 0.047 & 0.038 & 0.044 & 0.052 & 0.052 & 1.000 & 0.995 & 1.000 & 1.000 & 1.000 \\
					\hline
				\end{tabular}
		\end{center}
	\end{table}
}

\csubsection{Real Data Analysis}

In practical data analysis, we applied the RIT and BIT to two datasets to demonstrate its effectiveness in rare event scenarios: (1) a gene dataset on cancer-related transcriptional activity, and (2) a transaction dataset for detecting anomalous merchants. The test statistics, significance levels, and computation times were recorded. All analyses were conducted in R on a PC with a 2.40 GHz Intel Core i5-9300H processor and 16 GB of RAM. Continuous features were standardized to have a mean of zero and a standard deviation of one.

\textbf{Gene p53 Dataset.} This dataset relates to the p53 gene \citep{danziger2009predicting} (\url{https://archive.ics.uci.edu/ml/datasets/p53+Mutants}), which regulates cancer-suppressing proteins. The response variable $Y$ indicates p53 transcriptional activity, where $Y=1$ corresponds to active transcription (case) and $Y=0$ to inactive transcription (control). The dataset includes 5048 features: 4826 based on 2D electrostatic/surface properties and the remainder are 3D distance-based features. After removing missing data, the dataset contains 16,449 controls and 143 cases, resulting in significant imbalance.
We test the independence between the 5408-dimensional features and mutant p53 transcriptional activity using rescaled distance covariance ($T_{\rm dcov}$) and rescaled projection covariance ($T_{\rm IPcov}$). For  BIT, the sampling ratio is set to $s = (5, 20, 50)$, corresponding to control group sizes of \(n_1s = (715, 2860, 7150)\). Table \ref{tab:p53} shows that both 2D and 3D features are correlated with transcriptional activity. With $s=5$, we obtain results similar to those from the full sample, while the computational time is reduced to approximately 0.2\%.

{\begin{table}[htb]
			\begin{center}
				\caption{Test Statistics, p-Values, and Computational Time in the p53 Mutants Dataset independence test}
				\label{tab:p53}
				\renewcommand\arraystretch{1}
				\begin{tabular}{l|l|cccc}
					\hline \hline
					Feature Group&&RIT&$s = 5$&$s = 20$&$s = 50$\\
					\hline
					\multirow{6}{*}{2D}&$T_{\rm dcov ,S}$& 18.76 & 19.24 &  18.48 &  18.69\\
					&(p-value)& ($<$0.001)&($<$0.001)& ($<$0.001)& ($<$0.001)\\
					&time(sec)& 7906.5 &17.1  &238.6  &1472.0\\
					\cline{2-6}
					&$T_{\rm IPcov ,S}$ &0.30 & 0.31 & 0.30 & 0.30\\
					&(p-value)& ($<$0.001)& ($<$0.001)& ($<$0.001)& ($<$0.001)\\
					&time(sec) & 69424.5 & 171.7& 2214.4  & 13215.8\\
					\hline
					\multirow{6}{*}{3D}&$T_{\rm dcov ,S}$& 2.15 & 2.19 & 2.14 &  2.15\\
					&(p-value)& ($<$0.001)& ($<$0.001)& ($<$0.001)& ($<$0.001)\\
					&time(sec) & 587.4 & 1.0 & 17.4 &  111.1\\
					\cline{2-6}
					&$T_{\rm IPcov ,S}$ & 0.16 & 0.16 & 0.16 &  0.16\\
					&(p-value)& ($<$0.001)& ($<$0.001)& ($<$0.001)& ($<$0.001)\\
					&time(sec) & 4469.6 & 9.94 & 134.7 & 826.8\\
					\hline\hline
				\end{tabular}
				\renewcommand\arraystretch{1.5}
			\end{center}
		\end{table}
}

\textbf{Credit Dataset.} The \textit{Give Me Some Credit} dataset is a publicly available dataset on Kaggle (\url{https://www.kaggle.com/c/GiveMeSomeCredit/data}) that provides credit records along with individual features. Clients with serious delinquencies of 90 days or more in the past two years are defined as high-risk clients (cases with $Y=1$), while others are categorized as low-risk or normal clients (controls with $Y=0$). The dataset includes 10 relevant features, grouped into three categories: repayment ability, credit history, and property status. See Table \ref{tab:credit_intro} for detailed descriptions.
After removing missing data, the dataset contains 7,927 cases and 106,469 controls. The goal is to determine whether these feature groups are related to the clients' risk status. We apply $T_{\rm dcov}$ and $T_{\rm IPcov}$ to test the relationship between feature groups and client risk. As shown in Table \ref{tab:credit}, all feature groups are significantly correlated with default risk, with credit history being the most influential, followed by repayment ability. The BIT with $s=3$ achieves similar results to the full RIT but with only 7\% of the computational time, demonstrating its efficiency as established in Theorem \ref{thm:sample_U}.

{\begin{table}[htb]
			\begin{center}
				\caption{Test Statistics, p-Values, and Computational Time in the Credit Dataset independence test}
				\label{tab:credit}
				\renewcommand\arraystretch{1}
				\begin{tabular}{l|l|ccc}
					\hline \hline
					Feature Group &&RIT&$s = 3$&$s = 5$\\
					\hline
					\multirow{6}{*}{\parbox{2.5cm}{Repayment Ability}} & $T_{\rm dcov ,S}$ & 1.05 & 1.04 & 1.05 \\
					&(p-value) & ($<$0.001) & ($<$0.001) & ($<$0.001) \\
					&time(sec) & 161.9 & 12.6 & 27.4 \\
					\cline{2-5}
					& $T_{\rm IPcov ,S}$ & 0.62 & 0.62 & 0.62\\
					&(p-value) & ($<$0.001) & ($<$0.001) & ($<$0.001)\\
					&time(sec) & 14941.0 & 1123.5 & 2492.4\\
					\hline
					\multirow{6}{*}{\parbox{2.5cm}{Credit History}} & $T_{\rm dcov ,S}$ & 2.26 & 2.25 & 2.27 \\
					&(p-value) & ($<$0.001) & ($<$0.001) & ($<$0.001) \\
					&time(sec) & 139.9 & 10.5 & 23.4\\
					\cline{2-5}
					& $T_{\rm IPcov ,S}$ & 1.03 & 1.04 & 1.03\\
					&(p-value) & ($<$0.001) & ($<$0.001) & ($<$0.001)\\
					&time(sec) & 9013.0 & 807.7 & 1705.2\\
					\hline
					\multirow{6}{*}{\parbox{2.5cm}{Property Status}} & $T_{\rm dcov ,S}$ & 0.17 & 0.18 & 0.17 \\
					&(p-value) & ($<$0.001) & ($<$0.001) & ($<$0.001) \\
					&time(sec) & 141.6 & 10.6 & 23.3 \\
					\cline{2-5}
					& $T_{\rm IPcov ,S}$ &  0.10 & 0.10 & 0.10\\
					&(p-value) & ($<$0.001) & ($<$0.001) & ($<$0.001)\\
					&time(sec) & 14330.5 & 1131.2 & 2546.2\\
					\hline \hline
				\end{tabular}
				\renewcommand\arraystretch{1.5}
			\end{center}
		\end{table}
}

%
%

\csection{CONCLUDING REMARKS}

In this paper, we introduced the rescaled independence test framework for rare event scenarios and rigorously established its theoretical properties. Since the convergence rate depends primarily on the number of rare events,  we further theoretically explore the boosted independent test based on subsampling, which significantly reduces computational complexity while preserving the similar asymptotic properties. We investigate the theoretical results for both fixed and  high-dimensional settings and compare the local power of RIT and BIT. Then, we generalized the framework to accommodate multi-class responses. This extension broadens the applicability of RIT and BIT framework to a wider range of real-world scenarios.

We discuss three directions for future research in rare event analysis. First, exploring the theoretical properties of alternative sampling methods to enhance the efficiency and accuracy of rare event analysis is a valuable area of study. Second, developing independence tests for non-independent observations in network data is an intriguing research direction. Third, for large-scale datasets distributed across different organizations, investigating efficient test algorithms in federated learning with privacy preservation presents a promising and important topic.

\bibliographystyle{asa}
\bibliography{Mybib}

\newpage

\begin{center}
	{\bf\Large { Supplementary Material of ``Do more observations bring more information in rare events?''}}\\
	
\end{center}

\scsection{Appendix A. Useful Lemmas}
\renewcommand{\theequation}{A.\arabic{equation}}
\setcounter{equation}{0}

In this part, we establish useful tools for analyzing the covariance of $T$. It can be calculated that
\begin{equation}
	\begin{aligned}
		&\operatorname{Var}\left( T \right) = \binom{n_0}{m_0}^{-2} \binom{n_1}{m_1}^{-2}  \sum_{i,j,k,l} \mathrm{E} \left\lbrace  h\left(X_{i_1}^{(0)}, \cdots,X_{i_{m_0}}^{(0)} ;X_{j_1}^{(1)}, \cdots,X_{j_{m_1}}^{(1)}\right) \right. \\
		&~~~~~~~~~~~~~~~~~~~~~~~~~~~~~~~~~~~~~~~~~~~ \left.  h\left(X_{k_1}^{(0)}, \cdots,X_{k_{m_0}}^{(0)} ;X_{l_1}^{(1)}, \cdots,X_{l_{m_1}}^{(1)}\right) \right\rbrace \label{eq:Var_T} .
	\end{aligned}
\end{equation}
In order to analyze the covariance of $T$ in (\ref{eq:Var_T}), we need to consider the covariance of $h\left(X_{1}^{(0)}, \cdots,X_{m_0}^{(0)} ;X_{1}^{(1)}, \cdots,X_{m_1}^{(1)}\right)$ . By the definition of projection for the kernel function in \eqref{eq:projection_h_ab}, we have
\begin{equation}
	\begin{aligned}
		&\mathrm{E}\left\lbrace h_{a,b}\left( X_{1}^{(0)}, \cdots,X_{a}^{(0)} ;X_{1}^{(1)}, \cdots,X_{b}^{(1)}\right)  \right\rbrace \\
		= & \mathrm{E}\left[ \mathrm{E}\left\lbrace h\left(X_{1}^{(0)}, \cdots,X_{m_0}^{(0)} ;X_{1}^{(1)}, \cdots,X_{m_1}^{(1)}\right)  \mid X_{1}^{(0)}, \cdots,X_{a}^{(0)} ;X_{1}^{(1)}, \cdots,X_{b}^{(1)} \right\rbrace \right]\\
		= & \mathrm{E}\left\lbrace h\left(X_{1}^{(0)}, \cdots,X_{m_0}^{(0)} ;X_{1}^{(1)}, \cdots,X_{m_1}^{(1)}\right) \right\rbrace = 0 \label{eq:Ehab} .
	\end{aligned}
\end{equation}
Specifically, in the analysis of (\ref{eq:Var_T}), we need to consider the covariance of $h(\cdot)$, the covariance of projection $h_{a,b}(\cdot)$,  and the covariance between $h(\cdot)$ and $h_{a,b}(\cdot)$. To facilitate the proof of the asymptotic theoretical properties of $T$, we establish the following lemmas.

\begin{lemma}
	\label{lemma:cor_hh_rs}
	Define $\xi_{r,s}=\operatorname{Var} \{ h_{r,s} ( X_{1}^{(0)} , \cdots,X_{r}^{(0)};X_{1}^{(1)} , \cdots,X_{s}^{(1)} ) \}$. Then we have the following conclusions.
	\begin{enumerate}
		\item[(1)] {Covariance of $h(\cdot)$.} When  $\left| \left\lbrace i_1, \cdots, i_{m_0}\right\rbrace \cap \left\lbrace k_1, \cdots, k_{m_0}\right\rbrace \right| = r$  and \\ $\left| \left\lbrace j_1, \cdots, j_{m_1}\}\cap\{l_1, \cdots, l_{m_1}\right\rbrace \right| =s$, the covariance of $h(\cdot)$ could be calculated as
		\begin{equation}
			\label{eq:cor_hh_rs}
			\mathrm{E}\left\lbrace h\left(X_{i_1}^{(0)}, \cdots,X_{i_{m_0}}^{(0)} ;X_{j_1}^{(1)}, \cdots,X_{j_{m_1}}^{(1)}\right) h\left(X_{k_1}^{(0)},\cdots,X_{k_{m_0}}^{(0)} ;X_{l_1}^{(1)}, \cdots,X_{l_{m_1}}^{(1)}\right) \right\rbrace =\xi_{r,s}.
		\end{equation}
		\item[(2)] {Covariance of $h_{a,b}(\cdot)$.} When  $\big|\{i_1, \cdots, i_{a}\}\cap\{k_1, \cdots, k_{c}\}\big|=r$  and $\big|\{j_1, \cdots, j_{b}\}\cap\{l_1, \cdots, l_{d}\}\big|=s$, the covariance between $h_{a,b}(\cdot)$ and $h_{c,d}(\cdot)$ could be calculated as
		\begin{equation}
			\label{eq:cor_hab_hcd}
			\mathrm{E}\left\lbrace h_{a,b}\left( X_{1}^{(0)}, \cdots, X_{i_a}^{(0)};X_{1}^{(1)}, \cdots,X_{j_b}^{(1)}\right)  h_{c,d}\left(X_{1}^{(0)} ,\cdots,X_{k_{c}}^{(0)} ;X_{1}^{(1)}, \cdots,,X_{l_{d}}^{(1)}\right) \right\rbrace = \xi_{r,s}.
		\end{equation}
		
		\item[(3)] {Covariance between $h(\cdot)$ and $h_{a,b}(\cdot)$.}
		When  $\big|\{i_1, \cdots, i_{a}\}\cap\{k_1, \cdots, k_{m_0}\}\big|=r$  and $\big|\{j_1, \cdots, j_{b}\}\cap\{l_1, \cdots, l_{m_1}\}\big|=s$, the covariance between $h_{a,b}(\cdot)$ and $h(\cdot)$ could be calculated as
		\begin{equation}
			\label{eq:cor_hab_h}
			\mathrm{E}\left\lbrace h_{a,b}\left( X_{1}^{(0)} , \cdots,X_{i_a}^{(0)};X_{1}^{(1)} ,\cdots,X_{j_b}^{(1)}\right)  h\left(X_{1}^{(0)},\cdots,X_{k_{m_0}}^{(0)} ;X_{1}^{(1)} ,\cdots ,X_{l_{m_1}}^{(1)}\right) \right\rbrace = \xi_{r,s}.
		\end{equation}
		
	\end{enumerate}
	
\end{lemma}

\begin{proof}[Proof of lemma \ref{lemma:cor_hh_rs}:]
	To prove the conclusion in \eqref{eq:cor_hh_rs}, we define the observation set $\mathcal{X}_{i,1,m}^{(q)}=\big\{X_{i_1}^{(q)}, \cdots,X_{i_{m}}^{(q)}\big\}$ for $q=0,1$, where the subscript $i$ ranges from $i_1$ to $i_{m}$. When $\big|\{i_1, \cdots, i_{m_0}\}\cap\{k_1, \cdots, k_{m_0}\}\big|=r$  and $\big|\{j_1, \cdots, j_{m_1}\}\cap\{l_1, \cdots, l_{m_1}\}\big|=s$, we denote $\mathcal{X}_{i,1,{r}}^{(0)}=\mathcal{X}_{k,1,{r}}^{(0)}$, $\mathcal{X}_{j,1,{s}}^{(1)}=\mathcal{X}_{l,1,{s}}^{(1)}$ as the shared elements in the respective groups. Then we have
	\begin{equation*}
		\begin{aligned}
			& \mathrm{E}\left\lbrace h\left( \mathcal{X}_{i,1,r}^{(0)},\mathcal{X}_{i,{r+1},{m_0}}^{(0)} ; \mathcal{X}_{j,1,s}^{(1)} , \mathcal{X}_{j,s+1,m_1}^{(1)}\right)  h\left( \mathcal{X}_{k,1,r}^{(0)},\mathcal{X}_{k,r+1,m_0}^{(0)} ; \mathcal{X}_{l,1,s}^{(1)} , \mathcal{X}_{l,s+1,m_1}^{(1)}\right)\right\rbrace \\
			= & \mathrm{E} \left[ \mathrm{E}\left\lbrace h\left( \mathcal{X}_{i,1,r}^{(0)},\mathcal{X}_{i,r+1,m_0}^{(0)} ; \mathcal{X}_{j,1,s}^{(1)} , \mathcal{X}_{j,s+1,m_1}^{(1)}\right) \right.\right.  \\
			& \qquad \qquad\qquad  \qquad \qquad \left.  \left.  h\left( \mathcal{X}_{k,1,r}^{(0)},\mathcal{X}_{k,r+1,m_0}^{(0)} ; \mathcal{X}_{l,1,s}^{(1)} , \mathcal{X}_{l,s+1,m_1}^{(1)}\right) \mid \mathcal{X}_{i,1,r}^{(0)}, \mathcal{X}_{j,1,s}^{(1)} \right\rbrace \right]
		\end{aligned}
	\end{equation*}
	\begin{equation*}
		\begin{aligned}
			= & \mathrm{E} \left[ \mathrm{E}\left\lbrace h\left( \mathcal{X}_{i,1,r}^{(0)},\mathcal{X}_{i,r+1,m_0}^{(0)} ; \mathcal{X}_{j,1,s}^{(1)} , \mathcal{X}_{j,s+1,m_1}^{(1)}\right) \mid \mathcal{X}_{i,1,r}^{(0)}, \mathcal{X}_{j,1,s}^{(1)} \right\rbrace \right.  \\
			& \qquad \qquad \qquad \qquad \quad \left.  \mathrm{E} \left\lbrace   h\left( \mathcal{X}_{k,1,r}^{(0)},\mathcal{X}_{k,r+1,m_0}^{(0)} ; \mathcal{X}_{l,1,s}^{(1)} , \mathcal{X}_{l,s+1,m_1}^{(1)}\right) \mid \mathcal{X}_{i,1,r}^{(0)}, \mathcal{X}_{j,1,s}^{(1)} \right\rbrace \right]  \\
			= & \mathrm{E}\left\lbrace h_{r,s}^{2}\left(\mathcal{X}_{i,1,r}^{(0)}, \mathcal{X}_{j,1,s}^{(1)} \right) \right\rbrace  = \xi_{r,s},
		\end{aligned}
	\end{equation*}
	where the second equality holds because of conditional independence, and the final equality follows from the definition of $h_{r,s}$.

	For the conclusion in \eqref{eq:cor_hab_hcd}, we consider $\big|\{i_1, \cdots, i_{a}\}\cap\{k_1, \cdots, k_{c}\}\big|=r$  and $\big|\{j_1, \cdots, j_{b}\}\cap\{l_1, \cdots, l_{d}\}\big|=s$ (without loss of generality, we denote $\mathcal{X}_{i,1,r}^{(0)}=\mathcal{X}_{k,1,r}^{(0)}$, $\mathcal{X}_{j,1,s}^{(1)}=\mathcal{X}_{l,1,s}^{(1)}$ as the shared elements).
	\begin{equation*}
		\begin{aligned}
			& \mathrm{E}\left\lbrace h_{a,b}\left( \mathcal{X}_{i,1,r}^{(0)},\mathcal{X}_{i,r+1,a}^{(0)} ; \mathcal{X}_{j,1,s}^{(1)} , \mathcal{X}_{j,s+1,b}^{(1)}\right)  h_{c,d}\left( \mathcal{X}_{k,1,r}^{(0)},\mathcal{X}_{k,r+1,c}^{(0)} ; \mathcal{X}_{l,1,s}^{(1)} , \mathcal{X}_{l,s+1,d}^{(1)}\right)\right\rbrace \\
			= & \mathrm{E} \left[ \mathrm{E}\left\lbrace h\left( \mathcal{X}_{i,1,r}^{(0)},\mathcal{X}_{i,r+1,m_0}^{(0)} ; \mathcal{X}_{j,1,s}^{(1)} , \mathcal{X}_{j,s+1,m_1}^{(1)}\right) \mid \mathcal{X}_{i,1,r}^{(0)},\mathcal{X}_{i,r+1,a}^{(0)} , \mathcal{X}_{j,1,s}^{(1)} , \mathcal{X}_{j,s+1,b}^{(1)} \right\rbrace \right.  \\
			& \qquad \qquad \left.  \mathrm{E} \left\lbrace   h\left( \mathcal{X}_{k,1,r}^{(0)},\mathcal{X}_{k,r+1,m_0}^{(0)} ; \mathcal{X}_{l,1,s}^{(1)} , \mathcal{X}_{l,s+1,m_1}^{(1)}\right) \mid \mathcal{X}_{k,1,r}^{(0)},\mathcal{X}_{k,r+1,c}^{(0)} , \mathcal{X}_{l,1,s}^{(1)} , \mathcal{X}_{l,s+1,d}^{(1)} \right\rbrace \right]  \\
			= & \mathrm{E} \left[ \mathrm{E}\left\lbrace h\left( \mathcal{X}_{i,1,r}^{(0)},\mathcal{X}_{i,r+1,m_0}^{(0)} ; \mathcal{X}_{j,1,s}^{(1)} , \mathcal{X}_{j,s+1,m_1}^{(1)}\right)  h\left( \mathcal{X}_{k,1,r}^{(0)},\mathcal{X}_{k,r+1,m_0}^{(0)} ;\mathcal{X}_{l,1,s}^{(1)} , \mathcal{X}_{l,s+1,m_1}^{(1)}\right) \right. \right. \\
			& \qquad \qquad \qquad \qquad \left. \left. \mid \mathcal{X}_{i,1,r}^{(0)},\mathcal{X}_{i,r+1,a}^{(0)}\mathcal{X}_{k,r+1,c}^{(0)},\mathcal{X}_{j,1,s}^{(1)} , \mathcal{X}_{j,s+1,b}^{(1)}, \mathcal{X}_{l,s+1,d}^{(1)} \right\rbrace \right]= \xi_{r,s},
		\end{aligned}
	\end{equation*}
	where the first equality follows from the definition of $h_{a,b}$, the final equality is derived from (\ref{eq:cor_hh_rs}).
	
	For the conclusion in \eqref{eq:cor_hab_h}, we consider $\big|\{i_1, \cdots, i_{a}\}\cap\{k_1, \cdots, k_{m_0}\}\big|=r$  and $\big|\{j_1, \cdots, j_{b}\}\cap\{l_1, \cdots, l_{m_1}\}\big|=s$ (without loss of generality, we denote $\mathcal{X}_{i,1,{r}}^{(0)}=\mathcal{X}_{k,1,{r}}^{(0)}$, $\mathcal{X}_{j,1,{s}}^{(1)}=\mathcal{X}_{l,1,{s}}^{(1)}$ as the shared elements).
	\begin{equation*}
		\begin{aligned}
			& \mathrm{E}\left\lbrace h_{a,b}\left( \mathcal{X}_{i,1,r}^{(0)},\mathcal{X}_{i,r+1,a}^{(0)} ; \mathcal{X}_{j,1,s}^{(1)} , \mathcal{X}_{j,s+1,b}^{(1)}\right) h\left( \mathcal{X}_{k,1,r}^{(0)},\mathcal{X}_{k,r+1,m_0}^{(0)} ; \mathcal{X}_{l,1,s}^{(1)} , \mathcal{X}_{l,s+1,m_1}^{(1)}\right)\right\rbrace \\
			= & \mathrm{E} \left[ \mathrm{E}\left\lbrace h_{a,b}\left( \mathcal{X}_{i,1,r}^{(0)},\mathcal{X}_{i,r+1,a}^{(0)} ; \mathcal{X}_{j,1,s}^{(1)} , \mathcal{X}_{j,s+1,b}^{(1)}\right)  h\left( \mathcal{X}_{k,1,r}^{(0)},\mathcal{X}_{k,r+1,m_0}^{(0)} ; \mathcal{X}_{l,1,s}^{(1)} , \mathcal{X}_{l,s+1,m_1}^{(1)}\right) \right.\right.\\
			& \qquad \qquad \qquad \qquad \qquad \qquad  \left. \left.  \mid \mathcal{X}_{i,1,r}^{(0)},\mathcal{X}_{i,r+1,a}^{(0)},\mathcal{X}_{j,1,s}^{(1)} , \mathcal{X}_{j,s+1,b}^{(1)} \right\rbrace \right]\\
		\end{aligned}
	\end{equation*}
	\begin{equation*}
		\begin{aligned}
			= & \mathrm{E} \left[ h_{a,b}\left( \mathcal{X}_{i,1,r}^{(0)},\mathcal{X}_{i,r+1,a}^{(0)} ; \mathcal{X}_{j,1,s}^{(1)} , \mathcal{X}_{j,s+1,b}^{(1)}\right)  \right.  \\
			& \qquad \qquad \left.  \mathrm{E} \left\lbrace   h\left( \mathcal{X}_{k,1,r}^{(0)},\mathcal{X}_{k,r+1,m_0}^{(0)} ; \mathcal{X}_{l,1,s}^{(1)} , \mathcal{X}_{l,s+1,m_1}^{(1)}\right) \mid \mathcal{X}_{i,1,r}^{(0)},\mathcal{X}_{i,r+1,a}^{(0)},\mathcal{X}_{j,1,s}^{(1)} , \mathcal{X}_{j,s+1,b}^{(1)} \right\rbrace \right]  \\
			= & \mathrm{E} \left\lbrace h_{a,b}\left( \mathcal{X}_{i,1,r}^{(0)},\mathcal{X}_{i,r+1,a}^{(0)} ; \mathcal{X}_{j,1,s}^{(1)} , \mathcal{X}_{j,s+1,b}^{(1)}\right)  h_{r,s} \left( \mathcal{X}_{i,1,r}^{(0)}; \mathcal{X}_{j,1,s}^{(1)} \right)\right\rbrace   =\xi_{r,s} ,
		\end{aligned}
	\end{equation*}
	where the second equality follows from conditional independence, and the last equality is derived from the formula (\ref{eq:cor_hab_hcd}).	\end{proof}

Next, to analyze the asymptotic properties of $h_{0,2}(\cdot)$, we provide a lemma that can be found in Chapter 5 of \cite{serfling2009approximation} (pp. 194).
\begin{lemma}
	\label{lemma:second_order_U_asym}
	\citep{serfling2009approximation} For the zero-mean symmetric kernel function $h_{0,2}(X_1^{(1)}, X_2^{(1)})$, we have
	$$S_{n_1} = \frac{1}{n_1} \sum_{i \neq j} h_{0,2}\left( X_i^{(1)}, X_j^{(1)}\right)  \xrightarrow{d} \sum_{j=1}^{\infty} \lambda_j \left( \chi_{1j}^2 - 1\right) ,$$
	where $\chi_{1j}^2$s are independent $\chi_1^2$ variables and $\lambda_j$ is eigenvalues of the linear operator $A$ with $A g(x) = \int_{-\infty}^{\infty} h_{0,2} (x,y) g(y) dF(y)$, which maps a function $g(x)$ to another $A g(x)$.
\end{lemma}

\scsection{Appendix B. Detailed Discussion for Examples of RIT}
\renewcommand{\theequation}{B.\arabic{equation}}
\setcounter{equation}{0}

1. Discussion on the Rescaled Pearson $T_r$. On the one hand, it can be calculated that the Pearson correlation coefficient $\hat{r}^2$ is \begin{equation}
	\begin{aligned}
		\label{eq:pearson_def}
		\hat{r}^2
		&= \sqrt{\frac{n_1}{n}} \frac{\frac{1}{n_1} \sum_{i=1}^{n_1} X_i^{(1)}  - \frac{1}{n}\sum_{i=1}^{n} X_i }{\sqrt{\frac{1}{n} \sum_{i=1}^{n} (X_i)^2 - (\frac{1}{n} \sum_{i=1}^{n} X_i)^2} \sqrt{1 - \frac{n_1}{n}}}.
	\end{aligned}
\end{equation}
On the other hand, if a kernel function \( h(X^{(1)},X^{(0)}) = X^{(1)} - X^{(0)} \) is adopted, we can calculate that the RIT is
\begin{equation}
	\begin{aligned}
		\label{eq:Boosted Pearson using kernel}
		T_{r}  = \frac{1}{n_1} \sum_{i=1}^{n_1} X_i^{(1)}  - \frac{1}{n_0}\sum_{i=1}^{n_0} X_i^{(0)}.
	\end{aligned}
\end{equation}
Then it can be verified that the difference between the numerator of \eqref{eq:pearson_def} rescaled by a factor of $\left( n/n_1\right)^{1/2}$ and \eqref{eq:Boosted Pearson using kernel} is $(\overline{X}^{(0)} - \overline{X}^{(1)})  n_1/n $, which converges to 0 as $n_1/n \rightarrow 0$, where $\overline{X}^{(0)}$ and $\overline{X}^{(1)}$ denote the average of $X^{(0)}_i$ and $X^{(1)}_i$ respectively.

2. Discussion on the Rescaled Kendall's $T_{\tau}$. The Kendall's $\tau$ can be calculated as
\begin{equation}
	\begin{aligned}
		\label{eq:tau_def}
		\tau(X,Y)
		& = \frac{2}{n^2} \sum_{i = 1}^{n_1} \sum_{j = 1}^{n_0} \operatorname{sgn}\left( X_i^{(1)} - X_j^{(0)} \right).
	\end{aligned}
\end{equation}
If a kernel function \( h(X^{(1)},X^{(0)}) = \operatorname{sgn}(X^{(1)} - X^{(0)}) \) is adopted, the rescaled Kendall's $T_{\tau}$ is
\begin{equation}
	\begin{aligned}
		\label{eq:Boosted tau using kernel}
		T_{\tau} = \frac{1}{n_0 n_1} \sum_{i = 1}^{n_1} \sum_{j = 1}^{n_0} \operatorname{sgn}\left( X_i^{(1)} - X_j^{(0)} \right).
	\end{aligned}
\end{equation}
Then it can be verified that \eqref{eq:tau_def}, when rescaled by a factor of \(n^2(2 n_1 n_0)^{-1} \), is equal to \eqref{eq:Boosted tau using kernel}. Therefore, in Section 2.2, \eqref{eq:tau_with_U_form} is referred to as the rescaled Kendall's \( \tau \).

Then, we discuss the equivalence of Kendall's $\tau$ and Spearman's $\rho$ in rare event scenario. Specifically, it can be verified Spearman's $\rho$ correlation between $X$ and $Y$ is
\begin{equation}
	\label{eq:rho_def}
	\begin{aligned}	
		\widehat{\rho}(X,Y) & = n^{-3}\sum_{(i,j,k)}^n\mathrm{sgn}\left( X_{i}-X_{j}\right) \mathrm{sgn}\left( Y_{i}-Y_{k}\right)\\
		& = \frac{1}{n^3} \left\lbrace  n_1\sum_{i=1}^{n_0} \sum_{j=1}^{n_0} \mathrm{sgn}\left( X_{i}^{(0)} - X_{j}^{(0)}\right) + n_0 \sum_{i=1}^{n_1} \sum_{j=1}^{n_1} \mathrm{sgn}\left( X_{i}^{(1)} - X_{j}^{(1)}\right) \right. \\
		& \qquad \qquad \qquad \qquad \qquad \qquad+ \left. n \sum_{i=1}^{n_0} \sum_{j=1}^{n_1} \mathrm{sgn}\left( X_{i}^{(0)} - X_{j}^{(1)}\right)\right\rbrace .
	\end{aligned}
\end{equation}
The first two terms in formula (\ref{eq:rho_def}) are clearly zero, so only the last term plays a role in the test for independence. Thus, it is identical to the Kendall's $\tau$.

3. Discussion on the Rescaled Distance Covariance $T_{\rm dcov}$. The classical distance covariance can be defined as
\begin{equation*}
	\begin{aligned}	
		& \widehat{\operatorname{dcov}}^2(X , Y) = \widehat{S}_1(X,Y) + \widehat{S}_2(X,Y) - 2\widehat{S}_3(X,Y),\\
		& \widehat{S}_{1}(X , Y) = \frac1{n^2} \sum_{i=1}^n\sum_{j=1}^n\left\| X_{i}-X_{j}\right\|_2 \left\| Y_i-Y_j\right\|_2, \\
		& \widehat{S}_{2}(X , Y) = \frac1{n^2}\sum_{i=1}^n\sum_{j=1}^n\left\| X_{i}-X_{j}\right\|_2  \frac1{n^2}\sum_{i=1}^n\sum_{j=1}^n\left\| Y_i - Y_j\right\|_2, \\
		& \widehat{S}_{3}(X , Y) =  \frac1{n^3} \sum_{i=1}^n \sum_{j=1}^n \sum_{l=1}^n \left\| X_{i}-X_{l}\right\|_2 \left\| Y_j-{Y}_l\right\|_2.
	\end{aligned}
\end{equation*}

It can be further verified that in rare event data,
\begin{equation}
	\begin{aligned}
		\label{eq:dc_def}
		\hspace{-17em}
		&\widehat{S}_{1}(X , Y) = \frac2{n^2} \sum_{i=1}^{n_1} \sum_{j=1}^{n_0} \left\| X_{i}^{(1)}-X_{j}^{(0)}\right\|_2, \\
		&\widehat{S}_{2}(X , Y) = \frac{2 n_0 n_1} {n^4} \left( \sum_{i=1}^{n_0} \sum_{j=1}^{n_0} \left\| X_{i}^{(0)}-X_{j}^{(0)}\right\|_2  + 2 \sum_{i=1}^{n_0} \sum_{j=1}^{n_1} \left\| X_{i}^{(0)}-X_{j}^{(1)}\right\|_2 \right.  \\
		& \qquad \qquad \qquad    \qquad \qquad \qquad  \qquad \qquad \qquad  \qquad + \left.  \sum_{i=1}^{n_1} \sum_{j=1}^{n_1} \left\| X_{i}^{(1)}-X_{j}^{(1)}\right\|_2  \right),
	\end{aligned}
\end{equation}
\begin{equation*}
	\begin{aligned}
		&\widehat{S}_{3}(X , Y) =  \frac{n_0}{n^3} \sum_{i=1}^{n_1} \sum_{j=1}^{n_1}\left\| X_{i}^{(1)}-X_{j}^{(1)}\right\|_2  + \frac{1}{n^2}\sum_{i=1}^{n_0} \sum_{j=1}^{n_1}\left\| X_{i}^{(0)}-X_{j}^{(1)}\right\|_2  \\
		& \qquad \qquad \qquad \qquad   \qquad \qquad  \qquad \qquad \qquad  \qquad + \frac{n_1}{n^3}  \sum_{i=1}^{n_0} \sum_{j=1}^{n_0}\left\| X_{i}^{(0)}-X_{j}^{(0)}\right\|_2,
	\end{aligned}
\end{equation*}
\begin{equation*}
	\begin{aligned}	
		&\widehat{\operatorname{dcov}}^2(X , Y) = \frac{n_0^2 n_1^2}{n^4} \left(\frac{4}{n_0 n_1} \sum_{i=1}^{n_1} \sum_{j=1}^{n_0}\left\| X_{i}^{(1)}-X_{j}^{(0)}\right\|_2  - \frac{2}{n_0^2} \sum_{i=1}^{n_0} \sum_{j=1}^{n_0}\left\| X_{i}^{(0)}-X_{j}^{(0)}\right\|_2  \right. \\
		& \qquad \qquad \qquad  \qquad \qquad \qquad  \qquad \qquad \qquad  \qquad \left. - \frac{2}{n_1^2} \sum_{i=1}^{n_1} \sum_{j=1}^{n_1}\left\| X_{i}^{(1)}-X_{j}^{(1)}\right\|_2\right).
	\end{aligned}
\end{equation*}

When we take the kernel function as
\begin{equation*}
	\begin{aligned}	
		h_{\operatorname{dcov}}(X_1^{(0)},X_2^{(0)}, X_1^{(1)},X_2^{(1)}) &= \left\|X_1^{(1)} - X_1^{(0)}\right\|_2 + \left\|X_1^{(1)} - X_2^{(0)}\right\|_2 + \left\|X_2^{(1)} - X_1^{(0)}\right\|_2
	\end{aligned}
\end{equation*}
\begin{equation*}
	\begin{aligned}
		&+ \left\|X_2^{(1)} - X_2^{(0)}\right\|_2 - 2 \left\|X_1^{(0)} - X_2^{(0)}\right\|_2 - 2 \left\|X_1^{(1)} - X_2^{(1)}\right\|_2,
	\end{aligned}
\end{equation*}
the rescaled distance covariance $T_{\rm dcov}$ can be expressed as
\begin{equation}
	\begin{aligned}	
		\label{eq:dc with kernel}
		T_{\rm dcov} &= \frac{1}{n_1(n_1-1)n_0(n_0-1)}\left\lbrace  4(n_0-1)(n_1-1)\sum_{i=1}^{n_1} \sum_{j=1}^{n_0}\left\| X_{i}^{(1)}-X_{j}^{(0)}\right\|_2 \right. \\
		& \qquad \qquad \left. - 2 n_1(n_1-1) \sum_{i \neq j}^{n_1} \left\|X_i^{(0)} - X_j^{(0)} \right\|_2 - 2 n_0 (n_0-1) \sum_{i \neq j}^{n_0} \left\|X_i^{(1)} - X_j^{(1)} \right\|_2 \right\rbrace \\
		& = \frac{4}{n_0 n_1} \sum_{i=1}^{n_1} \sum_{j=1}^{n_0}\left\| X_{i}^{(1)}-X_{j}^{(0)}\right\|_2  - \frac{2}{n_0 (n_0-1)} \sum_{i=1}^{n_0} \sum_{j=1}^{n_0}\left\| X_{i}^{(0)}-X_{j}^{(0)}\right\|_2
	\end{aligned}
\end{equation}
\begin{equation*}
	\begin{aligned}	
		& \qquad \qquad  \qquad \qquad  \qquad \qquad \qquad  \qquad  - \frac{2}{n_1 (n_1-1)} \sum_{i=1}^{n_1} \sum_{j=1}^{n_1}\left\| X_{i}^{(1)}-X_{j}^{(1)}\right\|_2.
	\end{aligned}
\end{equation*}
It can be noted that the difference between \eqref{eq:dc_def} multiplied by a rescale factor of $n^4/n_0^2 n_1^2$, and \eqref{eq:dc with kernel} lies only in
\begin{equation*}
	\frac{2}{n_0^2 (n_0-1)} \sum_{i=1}^{n_0} \sum_{j=1}^{n_0}\left\| X_{i}^{(0)}-X_{j}^{(0)}\right\|_2 + 	\frac{2}{n_1^2 (n_1-1)} \sum_{i=1}^{n_1} \sum_{j=1}^{n_1}\left\| X_{i}^{(1)}-X_{j}^{(1)}\right\|_2,
\end{equation*}
which is a term that can be neglected as $n_1\rightarrow\infty$.

For the improved projection covariance, its computation is almost identical to that of distance covariance. Therefore, we omit the detailed proof.

\scsection{Appendix C. Proof of Theorems}
\renewcommand{\theequation}{C.\arabic{equation}}
\setcounter{equation}{0}

\scsubsection{Appendix C.1. Proof of Theorem \ref{thm:gen_U}}

We begin by proving the first part of the theoretical result in Theorem \ref{thm:gen_U}. We first calculate the leading term of the variance of $T$. Second, we establish the asymptotic normality of $T$.

Step 1. (Variance of $T$)
For the variance of $T$, we have
\begin{equation}
	\label{eq:varT}
	\begin{aligned}
		\operatorname{Var}\left( T\right) & = \binom{n_0}{m_0}^{-2} \binom{n_1}{m_1}^{-2} \sum_c \mathrm{E} \left\lbrace h \left(X_{i^0_{1}}^{(0)}, \cdots,X_{i^0_{m_0}}^{(0)} ;X_{i^1_{1}}^{(1)}, \cdots,X_{i^1_{m_1}}^{(1)}\right) \right. \\
		& \qquad  \qquad  \qquad   \qquad  \qquad  \qquad  \left. h\left(X_{j^0_{1}}^{(0)}, \cdots,X_{j^0_{m_0}}^{(0)} ;X_{j^1_{1}}^{(1)}, \cdots,X_{j^1_{m_1}}^{(1)}\right)\right\rbrace ,
	\end{aligned}
\end{equation}
where the sum is taken over all index combinations.
By \eqref{eq:cor_hh_rs} in Appendix A, $$\mathrm{E} \left\lbrace h \left(X_{i^0_{1}}^{(0)}, \cdots,X_{i^0_{m_0}}^{(0)} ;X_{i^1_{1}}^{(1)}, \cdots,X_{i^1_{m_1}}^{(1)}\right) h\left(X_{j^0_{1}}^{(0)}, \cdots,X_{j^0_{m_0}}^{(0)} ;X_{j^1_{1}}^{(1)}, \cdots,X_{j^1_{m_1}}^{(1)}\right)\right\rbrace = \xi_{r,s}$$ holds when $\left| \left\lbrace i^0_{1}, \cdots, i^0_{m_0}\right\rbrace \cap \left\lbrace j^0_1, \cdots, j^0_{m_0}\right\rbrace \right| = r$ and $\left| \left\lbrace i^1_1, \cdots, i^1_{m_1}\right\rbrace \cap \left\lbrace j^1_1, \cdots, j^1_{m_1}\right\rbrace \right| = s$. Therefore, the most frequently occurring non-zero terms in the sum are $\xi_{1,0}$ and $\xi_{0,1}$, with the respective counts being $\binom{n_0}{1}\binom{n_0-1}{m_0-1}\binom{n_0-m_0}{m_0-1} \binom{n_1}{m_1}\binom{n_1-m_1}{m_1}$ and $\binom{n_1}{1}\binom{n_1-1}{m_1-1}\binom{n_1-m_1}{m_1-1}$ $ \binom{n_0}{m_0}\binom{n_0-m_0}{m_0}$. Thus,
\begin{equation}
	\label{eq:var_T}
	\begin{aligned}
		\operatorname{Var}\left( T\right) &=\binom{n_0}{m_0}^{-2} \binom{n_1}{m_1}^{-2} \left\lbrace  \binom{n_0}{1}\binom{n_0-1}{m_0 -1}\binom{n_0-m_0}{m_0 -1} \binom{n_1}{m_1}\binom{n_1-m_1}{m_1 } \xi_{1,0} + \right. \\
		&\qquad \qquad \qquad \left. \binom{n_1}{1}\binom{n_1-1}{m_1 -1}\binom{n_1-m_1}{m_1 -1} \binom{n_0}{m_0}\binom{n_0-m_0}{m_0 } \xi_{0,1} \right\rbrace + O\left( \frac{1}{n_0 n_1}\right) \\
		& = \frac{m_0^2}{n_0} \xi_{1,0} + \frac{m_1^2}{n_1} \xi_{0,1} + O\left({n_0^{-1} n_1^{-1}}\right) = \frac{m_1^2}{n_1}\xi_{0,1} + O\left( n_0^{-1}\right),
	\end{aligned}
\end{equation}
where the last equality follows from the fact that $n_1/n_0\rightarrow 0$.

Step 2. (Asymptotic normality of $T$) To apply the Central Limit Theorem, we define $V = m_1 n_1^{-1} \sum_{i=1}^{n_1} h_{0,1}( X^{(1)}_i)$.
By the independence of $X^{(1)}_i$, the Central Limit Theorem (CLT) states that ${n_1}^{1/2} V \xrightarrow{d} N(0,m_1^2 \xi_{0,1})$. Therefore, it suffices to prove that as $n_1 \rightarrow \infty$, ${n_1}^{1/2}  (T-V) \xrightarrow{p} 0$. Since $\mathrm{E}\left\lbrace {{n_1}^{1/2} (T - V)}\right\rbrace  = 0$, it is sufficient to demonstrate that as $n_1 \rightarrow \infty$, $\mathrm{Var} \left\lbrace {n_1}^{1/2} (T - V) \right\rbrace  \rightarrow 0$. From \eqref{eq:var_T}, we have $\mathrm{Var} \left( {n_1}^{1/2} T\right) \rightarrow m_1^2 \xi_{0,1}$. By the definition of $\xi_{0,1}$, we have $\mathrm{Var}\left({n_1}^{1/2} V \right) = m_1^2 \xi_{0,1}$. It can be verified that $$
\begin{aligned}
	&\mathrm{Cov}\left( {n_1}^{1/2} T,{n_1}^{1/2} V\right)  = m_1 n_1 \binom{n_1}{m_1}^{-1} \binom{n_0}{m_0}^{-1} \mathrm{E} \left\lbrace h_{0,1}\left(  X_1^{(1)}\right)  T \right\rbrace \\
	=& m_1 n_1 \binom{n_1}{m_1}^{-1} \binom{n_1 - 1}{m_1 - 1} \mathrm{E} \left\lbrace  h_{0,1}\left(  X_1^{(1)}\right)   h\left(  X_1^{(0)},\cdots X_{m_0}^{(0)},X_1^{(1)} , \cdots ,X_{m_1}^{(1)}\right)  \right\rbrace \\
	=& m_1^2 \xi_{0,1},
\end{aligned}
$$
where the first equality follows from the symmetry of $T$ and $V$, and the last equality from formula \eqref{eq:cor_hab_h} in Appendix A. Then we have $\operatorname{Var}\{ {n_1}^{1/2}(T-V)\} = \operatorname{Var}\left( {n_1}^{1/2} T\right) +\operatorname{Var}\left( {n_1}^{1/2} V\right) -2\operatorname{Cov}\left( {n_1}^{1/2}  T, {n_1}^{1/2}  V\right)  \rightarrow 0$ and ${n_1}^{1/2} T \xrightarrow{d} N\left( 0,m_1^2 \xi_{0,1}\right) $.

Then we prove the second part of the theoretical result in Theorem \ref{thm:gen_U}. This approach is similar to the proof of the first part in Theorem \ref{thm:gen_U}, where we first calculate the variance of $T$, and then construct a statistic with a known asymptotic distribution to approximate $T$.

Step 1. (Variance of $T$) The variance of \( T \) is given by formula \eqref{eq:varT}. When $\xi_{0,1} = 0$, the most frequently occurring non-zero terms in the sum are $\xi_{1,0}$ and $\xi_{0,2}$ (since $n_1/n_0\rightarrow 0$), with the respective counts being $\binom{n_0}{1}\binom{n_0-1}{m_0 -1}\binom{n_0-m_0}{m_0 -1} \binom{n_1}{m_1}\binom{n_1-m_1}{m_1}$ and $\binom{n_0}{m_0}\binom{n_0-m_0}{m_0}\binom{n_1}{2} \binom{n_1-2}{m_0-2}\binom{n_1-m_0}{m_0-2}$. Thus,	
\begin{equation*}
	\begin{aligned}
		\operatorname{Var}\left( T\right) &=\binom{n_0}{m_0}^{-2} \binom{n_1}{m_1}^{-2} \left\lbrace  \binom{n_0}{1}\binom{n_0-1}{m_0 -1}\binom{n_0-m_0}{m_0 -1} \binom{n_1}{m_1}\binom{n_1-m_1}{m_1 } \xi_{1,0} + \right. \\
		&\qquad \qquad \qquad \left. \binom{n_0}{m_0}\binom{n_0-m_0}{m_0}\binom{n_1}{2} \binom{n_1-2}{m_0-2}\binom{n_1-m_0}{m_0-2} \xi_{0,2} \right\rbrace + O\left( \frac{1}{n_0 n_1}\right) \\
		& = \frac{m_0^2}{n_0} \xi_{1,0} + \frac{m_1^2 (m_1-1)^2}{2 n_1^2} \xi_{0,2} + O\left({n_0^{-1} n_1^{-1}}\right) = \frac{m_1^2 (m_1-1)^2}{2 n_1^2}\xi_{0,2} + o\left( n_1^{-2}\right) ,
	\end{aligned}
\end{equation*}
where the final equality can be derived from either $\xi_{1,0} = 0$ or ${n_1^2}/{n_0} \rightarrow 0$. Notably, in the final equality, when the condition $\xi_{1,0} = 0$ is satisfied, the last equality holds immediately; when the condition $n_1^{-2}n_0 \rightarrow 0$ holds, the first term will be dominated by the second one.

Step 2. (Asymptotic normality of $T$) Denote $V = m_1 (m_1-1) n_1^{-1} (n_1-1)^{-1} $ $ \sum_{j > i} h_{0,2}( X^{(1)}_i,X^{(1)}_j)$. In this context, since \(h_{0,2}(\cdot,\cdot)\) is a symmetric kernel function with zero mean, we apply Lemma \ref{lemma:second_order_U_asym} in Appendix A, leading to an asymptotic distribution $n_1 V \xrightarrow{d} m_1(m_1-1) \sum_{j=1}^\infty\lambda_j(\chi_{1j}^2-1)/2$. It suffices to prove that $n_1 (T - V) \xrightarrow{p} 0$. Since $\mathrm{E} \left\lbrace  n_1 (T - V) \right\rbrace  = 0$, it is sufficient to demonstrate that $\mathrm{Var} \left\lbrace  n_1 (T - V) \right\rbrace \rightarrow 0$. It is known that $\mathrm{Var} \left( n_1 T\right) \rightarrow m_1^2 (m_1-1)^2 \xi_{0,2} /2$. On the other hand, we have
$\mathrm{Var} \left(  n_1 V \right) = m_1^2 (m_1-1)^2 \left(n_1 -1 \right)^{-2}  \sum_{j > i} \sum_{l>k} \mathrm{E} \{ h_{0,2}( X^{(1)}_i,X^{(1)}_j)  h_{0,2}( X^{(1)}_k,X^{(1)}_l) \} = m_1^2 (m_1-1)^2 \left(n_1 -1 \right)^{-2}  \binom{n_1}{2}  \xi_{0,2} \rightarrow m_1^2 (m_1-1)^2 \xi_{0,2}/2$, where the second equality follows from formula \eqref{eq:cor_hab_hcd}. Finally,  for the covariance between $n_1 T$ and $n_1 V$,
\begin{equation}
	\label{eq:varify_cov}
	\begin{aligned}
		&\mathrm{Cov}(n_1 T,n_1 V) = n_1^2 \frac{m_1 (m_1-1)}{2} \mathrm{E} \left\lbrace T h_{0,2}\left( X_1^{(1)},X_2^{(1)}\right)  \right\rbrace \\
		= & n_1^2 \frac{m_1 (m_1-1)}{2} \binom{n_0}{m_0}^{-1} \binom{n_1}{m_1}^{-1} \\
		& \qquad  \sum \mathrm{E} \left\lbrace   h_{0,2}\left( X_1^{(1)},X_2^{(1)}\right)  h\left( X^{(0)}_{i_1},\cdots,X^{(0)}_{i_{m_0}} ; X^{(1)}_{i_1},\cdots,X^{(1)}_{i_{m_1}}\right) \right\rbrace\\
		= & n_1^2 \frac{m_1 (m_1-1)}{2} \binom{n_1}{m_1}^{-1} \binom{n_1-2}{m_1-2} \xi_{0,2} \rightarrow  m_1^2 (m_1-1)^2  \xi_{0,2} / 2.
	\end{aligned}
\end{equation}

The terms in the summation can be obtained directly using formula (\ref{eq:cor_hab_h}), leading to
$\operatorname{Var}\left\lbrace n_1 (T-V) \right\rbrace
=\operatorname{Var}\left( n_1 T\right) +\operatorname{Var}\left( n_1 V\right) -2\operatorname{Cov}\left( n_1  T, n_1  V\right)  \rightarrow 0$,
which means $n_1 T \xrightarrow{d} m_1(m_1-1) \sum_{j=1}^\infty\lambda_j(\chi_{1j}^2-1) /2$.

\scsubsection{Appendix C.2. Proof of Theorem \ref{thm:T_div_p}}

The proof of the second part in Theorem \ref{thm:gen_U} demonstrates that \(n_1  (T- V) \xrightarrow{p}0\). This means \(T\) could be approximated by \(V = m_1 (m_1 - 1) n_1^{-1} (n_1 - 1)^{-1} \sum_{i < j} h_{0,2}( X^{(1)}_i,X^{(1)}_j)\).  Therefore, we only need to prove the asymptotic normality of \(V\). This statement also appears in \cite{zheng1996consistent}, and we have adapted it to align with our notation. Without loss of generality, we can consider $m_1 = 2$. To apply the martingale central limit theorem, we construct a martingale difference array $Z_{i}=\sum_{j=1}^{i-1}h_{0,2}(X_i^{(1)},X_j^{(1)})$, and verify that it satisfies the following two conditions required by the martingale central limit theorem:
\begin{align}
	&s^{-4} \sum_{i=2}^{n_1} \mathrm{E}\left( Z_{i}^4\right)  \to 0, \label{eq:cond_1}\\
	&s^{-2} \sum_{i=2}^{n_1} \mathrm{E}\left( Z_{i}^2 \mid X_1^{(1)}, \ldots, X_{i-1}^{(1)}\right)  \xrightarrow{p} 1 , \label{eq:cond_2}
\end{align}
where $\widetilde{V} = \binom{n_1}{2}V = \sum_{i=2}^{n_1}Z_{i}$ and $s^2=\mathrm{E} (\widetilde{V}^2 )$.

Step 1. (Constructing the Martingale Difference Array)
We construct an array $Z_i=\sum_{j=1}^{i-1}h_{0,2}(X_i^{(1)},X_j^{(1)})$ for $2 \leqslant i \leqslant n_1$. Note that in the case of $\xi_{0,1} = 0$, we have $\mathrm{E}\{ h_{0,2} (X_i^{(1)},X_j^{(1)}) \mid X_j^{(1)} \} = h_{0,1}(X_j^{(1)}) = 0$, which implies that $\mathrm{E}(Z_{i} \mid X_1^{(1)},\cdots,X_{i-1}^{(1)}) = 0$. Further, we define $\mathcal{F}_i$ as the $\sigma$-field generated by $\{ X_1^{(1)},\cdots,X_{i-1}^{(1)}\}$.
Consequently, the array $\{(Z_i,\mathcal{F}_i),2 \leqslant i \leqslant n_1\}$ forms a martingale difference array, and $\{(U_i = \sum_{j=2}^{i} Z_{j},\mathcal{F}_i), 2\leqslant i \leqslant n_1\}$ is a martingale. We further define $\widetilde{V} = \binom{n_1}{2}V = \sum_{i=2}^{n_1}Z_{i}$ and $s^2=\mathrm{E}( \widetilde{V}^2) $, thus $U_{n_1} = \widetilde{V}$. From formula \eqref{eq:cor_hab_hcd} and $\mathrm{E}(Z_i Z_j) = 0$, we find $\mathrm{E}( Z_{i}^2 ) = (i-1) \mathrm{E}\{ h_{0,2}^2( X^{(1)}_1,X^{(1)}_2) \}$ and $s^2= n_1 (n_1-1) \mathrm{E}\{ h_{0,2}^2( X^{(1)}_1,X^{(1)}_2) \} / 2$.


Step 2. (Check Liapounov Condition in \eqref{eq:cond_1})
The fourth-order moment of $Z_i$ can be calculated as
$$
\begin{aligned} &\mathrm{E}\left( Z_{i}^4\right) =\sum_{j=1}^{i-1}\mathrm{E}\left\lbrace h_{0,2}^4\left(X_i^{(1)},X_j^{(1)}\right)\right\rbrace +3\sum_{1\leqslant j,k\leqslant i-1;j\neq k}\mathrm{E}\left\lbrace h_{0,2}^2\left(X_i^{(1)},X_j^{(1)}\right)h_{0,2}^2\left(X_i^{(1)},X_k^{(1)}\right) \right\rbrace \\
	&~~=(i-1)\mathrm{E}\left\lbrace h_{0,2}^4\left(X_1^{(1)},X_2^{(1)}\right)\right\rbrace + 3(i-1)(i-2)\mathrm{E}\left\lbrace h_{0,2}^2\left(X_1^{(1)},X_2^{(1)}\right) h_{0,2}^2\left(X_1^{(1)},X_3^{(1)}\right)\right\rbrace,
\end{aligned}$$
Then we can verify \eqref{eq:cond_1},
$$\begin{aligned}
	\sum_{i=2}^{n_1} \mathrm{E}\left(Z_{i}^4\right)& \lesssim {n_1}^2 \mathrm{E}\left\lbrace h_{0,2}^4\left(X_1^{(1)},X_2^{(1)}\right)\right\rbrace  + {n_1}^3 \mathrm{E}\left\lbrace h_{0,2}^2\left(X_1^{(1)},X_2^{(1)}\right)h_{0,2}^2\left(X_1^{(1)},X_3^{(1)}\right)\right\rbrace \\
	& \leqslant {n_1}^3 \mathrm{E}\left\lbrace h_{0,2}^4\left(X_1^{(1)},X_2^{(1)}\right)\right\rbrace ,
\end{aligned}$$
where $a \lesssim b$ denotes that there exists a constant $C$ such that $a \leqslant Cb$. The last inequality follows from the fact that $\mathrm{E}\{h_{0,2}^2(X_1^{(1)},X_2^{(1)})h_{0,2}^2(X_1^{(1)},X_3^{(1)})\} \leqslant \mathrm{E}\{h_{0,2}^4(X_1^{(1)},X_2^{(1)})\}$. Therefore, condition \eqref{eq:condition_clt} indicates that $s^{-4} \sum_{i=2}^{n_1} \mathrm{E}(Z_{i}^4) \to 0$, which confirms that the Lyapunov condition is satisfied.

Step 3. (Check Condition in \eqref{eq:cond_2})
Define $v_{i}^2 = \mathrm{E}(Z_{i}^{2} \mid \mathcal{F}_{i-1})$. To demonstrate $s^{-2} \sum_{i=2}^{n_1}v_i^2 \rightarrow_p 1$, we will show
\begin{align}
	&\mathrm{E}\left( \sum_{i=2}^{n_1}v_i^2\right) = s^2, \label{eq:cond_E}\\
	&s^{-4} \mathrm{E}\left( \sum_{i=2}^{n_1}v_i^2 - s^2\right)^2 \to 0. \label{eq:cond_Var}
\end{align}

\noindent
Condition \eqref{eq:cond_E} can be verified by $\mathrm{E}\left( \sum_{i=2}^{n_1}v_i^2\right)=\mathrm{E}\left\{ \sum_{i=2}^{n_1}\mathrm{E}(Z_i^2 \mid \mF_{i-1})\right\}=\sum_{i=2}^{n_1}\mathrm{E}(Z_i^2)=\mathrm{E}\{(\sum_{i=2}^{n_1}Z_i)^2\} = s^2$ due to $\mathrm{E}\left( Z_i Z_j\right) = 0$ for $i \neq j$.
To check condition \eqref{eq:cond_Var}, recall that $G(x,y) = \mathrm{E}\left\lbrace h_{0,2}(X^{(1)},x)h_{0,2}(X^{(1)},y)\right\rbrace$. Then, we have $G(X_j^{(1)},X_j^{(1)}) = \mathrm{E}\{ h^2_{0,2}(X_i^{(1)},X_j^{(1)}) \mid X_j^{(1)} \}$, and
$$\begin{aligned}
	v_{i}^2 & = \mathrm{E}\left( Z_{i}^{2}\mid X_{1}^{(1)},\cdots ,X_{i-1}^{(1)}\right)\\
	& = \sum_{j=1}^{i-1}\sum_{k=1}^{i-1} \mathrm{E}\left\lbrace h_{0,2}\left(X_{i}^{(1)},X_{j}^{(1)} \right) h_{0,2}\left(X_{i}^{(1)},X_{k}^{(1)} \right) \mid X_{1}^{(1)},\cdots ,X_{i-1}^{(1)} \right\rbrace \\
	& = 2\sum_{1 \leqslant j<k \leqslant i-1} G\left( X_j^{(1)},X_k^{(1)}\right)  + \sum_{j=1}^{i-1} G\left( X_j^{(1)},X_j^{(1)}\right) .
\end{aligned}$$

Next, we will consider the covariance between $v_{i_1}^2$ and $v_{i_2}^2$. If $ j_1=k_1= j_2=k_2$, we have $\mathrm{E}\{ G( X_{j_1}^{(1)},X_{k_1}^{(1)})  G( X_{j_2}^{(1)},X_{k_2}^{(1)}) \} = \mathrm{E} \{G^2( X_1^{(1)},X_1^{(1)}) \}$. If $j_1=j_2 < k_1=k_2 $, we have $\mathrm{E} \{ G( X_{j_1}^{(1)},X_{k_1}^{(1)})  G( X_{j_2}^{(1)},X_{k_2}^{(1)}) \} = \mathrm{E} \{G^2( X_1^{(1)},X_2^{(1)}) \}$ and if $j_1=k_1 < j_2=k_2 $, we have $\mathrm{E}\{ G( X_{j_1}^{(1)},X_{k_1}^{(1)})  G( X_{j_2}^{(1)},X_{k_2}^{(1)}) \} = [ \mathrm{E} \{G(X_1^{(1)},X_1^{(1)})\}] ^2$. Otherwise, $\mathrm{E}\{ G( X_{j_1}^{(1)},X_{k_1}^{(1)})  G( X_{j_2}^{(1)},X_{k_2}^{(1)}) \} = 0$. So the covariance between $v_{i_1}^2$ and $v_{i_2}^2$ (assuming $i_1 < i_2$) will be
$$\begin{aligned}
	& \mathrm{E}\left( v_{i_1}^2 v_{i_2}^2\right) = \sum_{j=1}^{i_1-1}\sum_{k=1}^{i_2-1} \left[ \mathrm{E}\left\lbrace G\left( X_1^{(1)},X_1^{(1)}\right) \right\rbrace \right] ^2 - \sum_{j=1}^{i_1-1} \left[ \mathrm{E}\left\lbrace G\left( X_1^{(1)},X_1^{(1)}\right) \right\rbrace \right] ^2 \\
	& ~~~~~~\qquad + 4 \sum_{1\leqslant j<k\leqslant i_1-1} \mathrm{E}\left\lbrace G^2 \left( X_1^{(1)},X_2^{(1)}\right) \right\rbrace  + \sum_{j=1}^{i_1-1} \mathrm{E}\left\lbrace G^2\left( X_1^{(1)},X_1^{(1)}\right) \right\rbrace  \\
	&=(i_1-1)(i_2-2)\left[ \mathrm{E}\left\lbrace G\left( X_1^{(1)},X_1^{(1)}\right) \right\rbrace \right] ^2 + 2(i_1-1)(i_1-2)\mathrm{E}\left\lbrace G^2\left( X_1^{(1)},X_2^{(1)}\right) \right\rbrace \\
	&~~~~~\qquad + (i_1-1) \mathrm{E} \left\lbrace G^2\left( X_1^{(1)},X_1^{(1)}\right) \right\rbrace .
\end{aligned}$$
Then, we consider each component of
$$\mathrm{E}\left( \sum_{i=2}^{n_1}v_i^2 - s^2\right)^2 = \mathrm{E}\left\lbrace \left( \sum_{i=2}^{n_1} v_i^2\right) ^2\right\rbrace - s^4  =2 \sum_{2 \leqslant i<j\leqslant n_1} \mathrm{E}\left( v_i^2 v_j^2 \right) + \sum_{i=2}^{n_1} \mathrm{E}\left( v_i^4\right) - s^4,$$
with
$$\begin{aligned}
	\sum_{2 \leqslant i<j\leqslant n_1} \mathrm{E}\left( v_i^2 v_j^2 \right) &= \frac{1}{8} \left\lbrace  n_1^4 + O \left( n_1^3\right) \right\rbrace  \left[ \mathrm{E}\left\lbrace G\left( X_1^{(1)},X_1^{(1)}\right) \right\rbrace \right] ^2  \\
	& + \frac{1}{6} \left\lbrace n_1^4 + O \left( n_1^3\right) \right\rbrace \mathrm{E}\left\lbrace G^2\left( X_1^{(1)},X_2^{(1)}\right) \right\rbrace
	+ O \left( n_1^3\right) \mathrm{E} \left\lbrace G^2\left( X_1^{(1)},X_1^{(1)}\right)\right\rbrace, \\
	\sum_{i=2}^{n_1} \mathrm{E}\left( v_i^4\right) &= O\left( n_1^3\right)  \mathrm{E} \left\{G^2\left( X_1^{(1)},X_2^{(1)}\right) \right\} + O\left(n_1^2\right)  \mathrm{E} \left\{G^2\left( X_1^{(1)},X_1^{(1)}\right) \right\}, \\
	s^4 &= \frac{1}{4} \left\lbrace  n_1^4 + O \left( n_1^3\right) \right\rbrace  \left[ \mathrm{E}\left\lbrace G\left( X_1^{(1)},X_1^{(1)}\right)\right\rbrace \right] ^2,
\end{aligned}$$
where the last equality is derived from the fact $\mathrm{E}\{G(X_1^{(1)},X_1^{(1)})\} = \mathrm{E}\{h_{0,2}^2( X^{(1)}_1,X^{(1)}_2) \} $. This leads to
\begin{equation*} \begin{aligned}
		&\mathrm{E}\left( \sum_{i=2}^{n_1}v_i^2 - s^2\right)^2\\
		\lesssim & n_1^3 \left[ \mathrm{E}\left\lbrace G\left(X_1^{(1)},X_1^{(1)}\right)\right\rbrace \right] ^2 + n_1^4 \mathrm{E}\left\lbrace G^2\left(X_1^{(1)},X_2^{(1)}\right)\right\rbrace + n_1^3 \mathrm{E} \left\lbrace G^2\left(X_1^{(1)},X_1^{(1)}\right)\right\rbrace \\
		\leqslant   &  n_1^4 \mathrm{E}\left\lbrace G^2\left(X_1^{(1)},X_2^{(1)}\right)\right\rbrace + n_1^3 \mathrm{E} \left\lbrace G^2\left(X_1^{(1)},X_1^{(1)}\right)\right\rbrace\\
		\leqslant &n_1^4 \mathrm{E}\left\lbrace	G^2\left(X_1^{(1)},X_2^{(1)}\right)\right\rbrace + n_1^3 \mathrm{E} \left\lbrace h_{0,2}^4\left(X_1^{(1)},X_2^{(1)}\right)\right\rbrace   .
	\end{aligned}
\end{equation*}
Note that in the first inequality, the terms containing $n_1^4 [ \mathrm{E}\{G(X_1^{(1)},X_1^{(1)})\}] ^2$ are eliminated, leaving only the lower-order term of $n_1^3$. The second inequality follows from $[ \mathrm{E}\{G(X_1^{(1)},X_1^{(1)})\}]^2 \leqslant \mathrm{E} \{ G^2(X_1^{(1)},X_1^{(1)})\}$, and the last inequality is derived from\\ $\mathrm{E} \{ G^2(X_1^{(1)},X_1^{(1)}) \} = \mathrm{E} ( [\mathrm{E} \{ h^2_{0,2}(X_1^{(1)},X_2^{(1)}) \mid X_2^{(1)} \} ]^2 ) \leqslant \mathrm{E} [ \mathrm{E} \{ h^4_{0,2}(X_1^{(1)},X_2^{(1)}) \mid X_2^{(1)} \}] =  \mathrm{E} \{ h_{0,2}^4(X_1^{(1)},X_2^{(1)})\}$.

Therefore, condition \eqref{eq:condition_clt} indicates that \eqref{eq:cond_Var} is satisfied. Conditions \eqref{eq:cond_E} and \eqref{eq:cond_Var} together imply that \eqref{eq:cond_2} holds. We can then apply martingale central limit theorem to conclude that $n_1 V / \xi_{0,2}^{1/2} \xrightarrow{d} N(0,2)$. When $m_1 > 2$, it can be verified that its variance is ${m_1^2 (m_1-1)^2}/{2}$.

\scsubsection{Appendix C.3. Proof of theorem \ref{thm:sample_U}}

We first prove the first part of the theoretical result in  Theorem \ref{thm:sample_U}. The proof steps are similar to the proof of Theorem \ref{thm:gen_U} by calculating the variance of $T_S$ and establishing its asymptotic normality.


Step 1. (Variance of $T_{\rm S}$)
The variance of $T_{\rm S}$ has the following form:
\begin{equation*}
	\begin{aligned}
		\mathrm{Var}\left(T_{\rm S}\right)  &= \binom{s n_1}{m_0}^{-2} \binom{n_1}{m_1}^{-2} \sum_c \mathrm{E}\left\lbrace   \left(\prod_{j=1}^{m_0}\delta_{i^0_j}\right) h\left(X_{i^0_1}^{(0)}, \cdots,X_{i^0_{m_0}}^{(0)} ;X_{i^1_{1}}^{(1)}, \cdots,X_{i^1_{m_1}}^{(1)}\right)\right.
	\end{aligned}
\end{equation*}
\begin{equation}
	\label{eq:varTS}
	\begin{aligned}
		\qquad \qquad \qquad \qquad \qquad  \left. \left(\prod_{j=1}^{m_0}\delta_{l^0_j}\right) h\left(X_{l^0_1}^{(0)}, \cdots,X_{l^0_{m_0}}^{(0)} ;X_{l^1_{1}}^{(1)}, \cdots,X_{l^1_{m_1}}^{(1)}\right) \right\rbrace .
	\end{aligned}
\end{equation}
where the sum is taken over all index combinations. The expectation \begin{equation*}
	\begin{aligned}
		&\mathrm{E}\left\lbrace   \left(\prod_{j=1}^{m_0}\delta_{i^0_j}\right) h\left(X_{i^0_1}^{(0)}, \cdots,X_{i^0_{m_0}}^{(0)} ;X_{i^1_{1}}^{(1)}, \cdots,X_{i^1_{m_1}}^{(1)}\right)\right. \\
		& \qquad \left. \left(\prod_{j=1}^{m_0}\delta_{l^0_j}\right) h\left(X_{l^0_1}^{(0)}, \cdots,X_{l^0_{m_0}}^{(0)} ;X_{l^1_{1}}^{(1)}, \cdots,X_{l^1_{m_1}}^{(1)}\right)\right\rbrace  = \left( \frac{sn_1}{n_0} \right)^{2m_0-r} \xi_{r,s}
	\end{aligned}
\end{equation*}
holds when $\left| \left\lbrace i^0_{1}, \cdots, i^0_{m_0}\right\rbrace \cap \left\lbrace l^0_1, \cdots, l^0_{m_0}\right\rbrace \right| = r$ and $\left| \left\lbrace i^1_1, \cdots, i^1_{m_1}\right\rbrace \cap \left\lbrace l^1_1, \cdots, l^1_{m_1}\right\rbrace \right| = s$. Here, the term$(sn_1/n_0 )^{2m_0-r}$ arises because the expectation involves $2m_0-r$ independent $\delta$'s, and $\delta$ is independent of the observations $X$. Consequently, we can treat the two parts separately when calculating the expectation. Therefore, the variance of $T_{\rm S}$ can be calculated by counting the cases mentioned above:
\begin{equation*}
	\begin{aligned}
		\operatorname{Var}\left( T_{\rm S} \right) &=\binom{s n_1}{m_0}^{-2} \binom{n_1}{m_1}^{-2} \\
		& \quad \left\lbrace  \binom{n_0}{1}\binom{n_0-1}{m_0 -1}\binom{n_0-m_0}{m_0 -1} \binom{n_1}{m_1}\binom{n_1-m_1}{m_1 } \left( \frac{s n_1}{n_0}\right) ^{2m_0-1}\xi_{1,0} + \right. \\
		& \quad \left. \binom{n_1}{1}\binom{n_1-1}{m_1 -1}\binom{n_1-m_1}{m_1 -1} \binom{n_0}{m_0}\binom{n_0-m_0}{m_0} \left( \frac{s n_1}{n_0}\right) ^{2m_0} \xi_{0,1} \right\rbrace + O\left( \frac{1}{n_1^2}\right) \\
		& = \frac{m_0^2}{s n_1} \xi_{1,0} + \frac{m_1^2}{n_1} \xi_{0,1} + O\left(n_1^{-2}\right).
	\end{aligned}
\end{equation*}

Step 2. (Asymptotic normality of $T_{\rm S}$) Next, we utilize the central limit theorem to demonstrate that $T_{\rm S}$ possesses asymptotic normality. Let $V_{\rm S} = m_1 n_1^{-1} \sum_{i=1}^{n_1} h_{0,1}( X_i^{(1)})  + m_0 (s n_1)^{-1} \sum_{i=1}^{n_0} \delta_i h_{1,0}( X_i^{(0)}) $. By the independence of $X_i^{(1)}$, $X_i^{(0)}$, and  $\delta_i$, along with the central limit theorem, we know that
\begin{equation*}
	\frac{m_1}{{n_1}^{1/2}}  \sum_{i=1}^{n_1} h_{0,1}\left( X_i^{(1)}\right)  \xrightarrow{d} N \left( 0,m_1^2 \xi_{0,1}\right)
\end{equation*}
\begin{equation}
	\label{eq:CLT_for_sample_countrol}
	\frac{m_0}{s {n_1}^{1/2}}  \sum_{i=1}^{n_0} \delta_i h_{1,0} \left(  X_i^{(0)}\right)   \xrightarrow{d} N\left( 0, \frac{m_0^2}{s}\xi_{1,0}\right),
\end{equation}
which implies ${n_1}^{1/2}V_{\rm S} \xrightarrow{d} N \left( 0, s^{-1} m_0^2 \xi_{1,0} + m_1^2\xi_{0,1}\right)$. Here, \eqref{eq:CLT_for_sample_countrol} follows from the \\
Lyapunov condition: $[\mathrm{Var}\{\sum_{i=1}^{n_0} \delta_{i} h_{1,0}(X_i^{(0)})\}]^{-2} \sum_{i=1}^{n_0} \mathrm{E}\{\delta_{i}^4 h_{1,0}^4(X_i^{(0)}) \} = (sn_1 \xi_{0,1})^{-2} $\\$ s n_1 \mathrm{E}\{h_{1,0}^4(X^{(0)}) \} \rightarrow 0$.  Therefore, it suffices to prove that as $n_1 \to \infty$, ${n_1}^{1/2} \left( T_{\rm S}-V_{\rm S}\right)  \xrightarrow{p} 0$. Since $\mathrm{E}\left\lbrace {n_1}^{1/2} \left( T_{\rm S}- V_{\rm S} \right) \right\rbrace  = 0$, it is sufficient to demonstrate that $\mathrm{Var}\left\lbrace  {n_1}^{1/2} \left( T_{\rm S}- V_{\rm S} \right) \right\rbrace  \to 0$. On the one hand, $\mathrm{Var} \left(  {n_1}^{1/2} T_{\rm S}\right)  \to s^{-1} m_0^2 \xi_{1,0} + m_1^2 \xi_{0,1}$. On the other hand, we have $
\mathrm{Var}\left( {n_1}^{1/2} V_{\rm S} \right)  = s^{-1} m_0^2 \xi_{1,0} + m_1^2 \xi_{0,1}$. Moreover, the covariance can be computed by
\begin{equation*}
	\begin{aligned}
		&\mathrm{Cov}\left( {n_1}^{1/2} T_{\rm S},{n_1}^{1/2} V_{\rm S}\right)   \\
		= & n_1 \left[ m_1 \mathrm{E} \left\lbrace  T_{\rm S}h_{0,1}\left(  X_1^{(1)}\right)  \right\rbrace   + m_0 \frac{n_0}{s n_1} \mathrm{E} \left\lbrace  T_{\rm S}\delta_1 h_{1,0}\left( X_1^{(0)}\right) \right\rbrace  \right]  \\
		= & m_1 n_1 \binom{s n_1}{m_0}^{-1} \binom{n_1}{m_1}^{-1} \binom{n_1-1}{m_1-1} \binom{n_0}{m_0} \\
		& \qquad \mathrm{E} \left\lbrace  h_{0,1}\left(  X_1^{(1)}\right)  \left(\prod_{i=1}^{m_0} \delta_{i} \right) h\left(  X_1^{(0)},\cdots,X_{m_0}^{(0)};X_1^{(1)},\cdots,X_{m_1}^{(1)}\right) \right\rbrace  \\
		& \quad + m_0 n_1 \frac{n_0}{sn_1} \binom{s n_1}{m_0}^{-1} \binom{n_1}{m_1}^{-1} \binom{n_0-1}{m_0-1} \binom{n_1}{m_1} \\
		& \qquad \mathrm{E} \left\lbrace \delta_{1} h_{0,1}\left(  X_1^{(1)}\right) \left(\prod_{i=1}^{m_0} \delta_{i} \right)  h\left(  X_1^{(0)},\cdots,X_{m_0}^{(0)};X_1^{(1)},\cdots,X_{m_1}^{(1)}\right) \right\rbrace  \\
		= & m_1^2 \xi_{0,1} + \frac{m_0^2}{s} \xi_{1,0}.
	\end{aligned}
\end{equation*}
Then we have $\operatorname{Var}\left\lbrace {n_1}^{1/2}(T_{\rm S}-V_{\rm S})\right\rbrace  \to 0 $, implying ${n_1}^{1/2} T_{\rm S} \xrightarrow{d} N( 0, s^{-1} m_0^2 \xi_{1,0} + m_1^2 \xi_{0,1}) $.

Next, we prove the second part in Theorem \ref{thm:sample_U}. We focus on the calculation of its variance. Specifically, the variance of $T_{\rm S}$ retains the form given in \eqref{eq:varTS}, and we only need to recalculate it under the condition that $\xi_{0,1} = \xi_{1,0} = 0$:
\begin{equation*}
	\begin{aligned}
		\operatorname{Var}\left( T_{\rm S}\right) = \binom{s n_1}{m_0}^{-2} \binom{n_1}{m_1}^{-2} \left\lbrace   \binom{n_1}{2}\binom{n_1-2}{m_1-2} \binom{n_1-m_1}{m_1-2} \binom{n_0}{m_0} \binom{n_0-m_0}{m_0} \left(\frac{s n_1}{n_0} \right)^{2m_0} \xi_{0,2} \right.
	\end{aligned}
\end{equation*}
\begin{equation}
	\label{eq:varTSresult}
	\begin{aligned}
		&+  \binom{n_1}{1}\binom{n_1-1}{m_1-1}\binom{n_1-m_1}{m_1-1}\binom{n_0}{1}\binom{n_0-1}{m_0-1}\binom{n_0-m_0}{m_0-1} \left(\frac{s n_1}{n_0} \right)^{2m_0-1} \xi_{1,1} \\
		&+ \left.  \binom{n_1}{m_1} \binom{n_1-m_1}{m_1} \binom{n_0}{2} \binom{n_0-2}{m_0-2} \binom{n_0-m_0}{m_0-2} \left( \frac{s n_1}{n_0} \right)^{2m_0-2} \xi_{2,0} \right\rbrace	+o\left(n_1^{-2}\right) \\
		&=  \frac{ m_1^2(m_1-1)^2}{2 n_1^2} \xi_{0,2} + \frac{m_1^2 m_0^2}{s n_1^2}\xi_{1,1} + \frac{m_0^2(m_0-1)^2}{2 s^2 n_1^2} \xi_{2,0} + o\left(n_1^{-2} \right).
	\end{aligned}
\end{equation}
It could be verified that $T_{\rm S}$ can be approximated by $V_{\rm S} = m_1(m_1-1) n_1^{-1} (n_1-1)^{-1} \sum_{i < j} h_{0,2}( X_i^{(1)},X_j^{(1)})  + m_0 m_1 (s n_1)^{-1}  n_1^{-1} \sum_{i,j} \delta_i h_{1,1}( X_i^{(0)}, X_j^{(1)}) + m_0 (m_0-1) (s n_1)^{-1}\\ (s n_1 - 1)^{-1} \sum_{i<j} \delta_i \delta_j h_{2,0}( X_i^{(0)}, X_j^{(0)})$. Notably, $h_{2,0}( X_i^{(0)}, X_j^{(1)})$, $h_{1,1}( X_i^{(0)}, X_j^{(1)})$ and $h_{0,2}( X_i^{(0)}, X_j^{(1)})$ are dependent. Thus, it is hard to find an martingale difference sequence to denote $V_{\rm S}$. Consequently, we only focus on the variance here.

\scsubsection{Appendix C.4. Proof of Theorem \ref{thm:sample_div_p}}

In the proof of Theorem \ref{thm:T_div_p}, we have demonstrated that the asymptotic distribution of $V = m_1(m_1-1)n_1^{-1}(n_1-1)^{-1}  \sum_{i < j} h_{0,2}\left( X^{(1)}_i,X^{(1)}_j\right)$ is a normal distribution with mean 0 and variance $m_1^2 (m_1-1)^2 \xi_{0,2} /2$. Next, we will prove that as $s \to \infty$, $T_{\rm S}$ can be well approximated by $V$, meaning that $n_1 \left( T_{\rm S} - V \right) \xrightarrow{p} 0 $. Since $\mathrm{E}\left\lbrace n_1  \left( T_{\rm S}- V \right) \right\rbrace  = 0$, it is sufficient to demonstrate that $\mathrm{Var}\left\lbrace  n_1  \left( T_{\rm S}- V \right) \right\rbrace  \to 0$. It is known that $\mathrm{Var} \left( n_1 T_{\rm S}\right)  \to m_1^2 (m_1-1)^2 \xi_{0,2} /2$ from formula \eqref{eq:varTSresult}, and we have $
\mathrm{Var}\left( n_1 V \right) \to m_1^2 (m_1-1)^2 \xi_{0,2} /2$. For the covariance, it can be computed as
\begin{equation*}
	\begin{aligned}
		&\mathrm{Cov}(n_1 T_{\rm S},n_1 V) = n_1^2 \frac{m_1 (m_1-1)}{2} \mathrm{E} \left\lbrace T_{\rm S} h_{0,2}\left( X_1^{(1)},X_2^{(1)}\right)  \right\rbrace \\
		= & n_1^2 \frac{m_1 (m_1-1)}{2} \binom{s n_1}{m_0}^{-1} \binom{n_1}{m_1}^{-1} \\
		& \qquad \qquad \sum_{c} \mathrm{E} \left\lbrace   h_{0,2}\left( X_1^{(1)},X_2^{(1)}\right) \left(\prod_{j=1}^{m_0}\delta_{i^0_j}\right)  h\left( X^{(0)}_{i_1},\cdots,X^{(0)}_{i_{m_0}} ; X^{(1)}_{i_1},\cdots,X^{(1)}_{i_{m_1}}\right) \right\rbrace\\
	\end{aligned}
\end{equation*}
\begin{equation*}
	\begin{aligned}
		= & n_1^2 \frac{m_1 (m_1-1)}{2} \binom{s n_1}{m_0}^{-1} \binom{n_1}{m_1}^{-1} \\
		& \sum_c \mathrm{E} \left\lbrace  h_{0,2}\left( X_1^{(1)},X_2^{(1)}\right)  \left(\prod_{j=1}^{m_0}\delta_{i^0_j}\right) h\left( X^{(0)}_{i_{1}},\cdots,X^{(0)}_{i_{m_0}} ; X^{(1)}_1,X^{(1)}_2,X^{(1)}_{i_{3}},\cdots,X^{(1)}_{i_{m_1}}\right) \right\rbrace\\
		= & n_1^2 \frac{m_1 (m_1-1)}{2} \binom{s n_1}{m_0}^{-1} \binom{n_1}{m_1}^{-1} \binom{n_0}{m_0} \binom{n_1-2}{m_1-2} \left( \frac{sn_1}{n_0}\right)^{m_0} \xi_{0,2} \rightarrow \frac{m_1^2 (m_1-1)^2}{2} \xi_{0,2} .
	\end{aligned}
\end{equation*}
So we can obtain
$\operatorname{Var}\left\lbrace n_1 (T_{\rm S}-V) \right\rbrace
\to 0$, meaning that $n_1 T_{\rm S} \xrightarrow{d} n_1 V$.

\scsubsection{Appendix C.5. Proof of Theorem \ref{thm:multiple_RIT}}

Since Theorem \ref{thm:multiple_RIT} generalizes binary $Y$ to multi-class $Y$, its proof is similar to that of Theorem \ref{thm:gen_U}. Since the proof of asymptotic normality is straightforward by consider $V_K = \sum_{k=1}^K m_k n_k^{-1} \sum_{i=1}^{n_k} \underbrace{h_{0,\cdots,0,1,0,\cdots,0}}_{\text{only } \nu_k \text{ is 1}} (X_i^{(k)})$, we focus only on the calculation of the variance of $T$ in both of the two cases. The variance of $T_{K}$ has the following form: 	
\begin{equation*}
	\begin{aligned}
		\mathrm{Var}\left( T_{K}\right)  &= \prod_{j=0}^K \binom{n_j}{m_j}^{-2}  \sum_c \mathrm{E} \left\lbrace h\left(X_{i^0_{1}}^{(0)}, \cdots,X_{i^0_{m_0}}^{(0)} ; \cdots ;X_{i^K_{1}}^{(K)}, \cdots,X_{i^K_{m_K}}^{(K)}\right)\right. \\
		& \qquad \qquad \qquad \qquad \qquad \left. h\left(X_{j^0_{1}}^{(0)}, \cdots,X_{j^0_{m_0}}^{(0)} ; \cdots ;X_{j^K_{1}}^{(K)}, \cdots,X_{j^K_{m_K}}^{(K)}\right)\right\rbrace.
	\end{aligned}
\end{equation*}

In the first part of Theorem \ref{thm:multiple_RIT}, we analyze the most frequently occurring components of $\mathrm{Var}\left( T_{K}\right)$ to derive the leading term of the variance:
\begin{equation*}
	\begin{aligned}
		\mathrm{Var}\left( T_{K}\right)  &= \left\lbrace \prod_{j=0}^K \binom{n_j}{m_j}^{-2}\right\rbrace \left[  \sum_{k=1}^K \binom{n_k}{1}\binom{n_k-1}{m_k -1}\binom{n_k-m_k}{m_k -1} \right.  \\
		& \qquad \qquad  \left. \prod_{l \neq k} \left\lbrace \binom{n_l}{m_l}\binom{n_l-m_l}{m_l}\right\rbrace  \zeta_{1,k}   \right] + o\left(n_1^{-1} \right) \rightarrow \sum_{k=1}^K \frac{m_k^2}{n_k} \zeta_{1,k}.
	\end{aligned}
\end{equation*}
Together with the assumption $n_k/n_1 \rightarrow r_k$, we can proof the first part of Theorem \ref{thm:multiple_RIT} by verifying $n_1^{1/2} \left(T_{K}-V_K \right) \xrightarrow{p} 0$. The specific details of the verifications are similar to the techniques presented in \eqref{eq:varify_cov}.

In the second part of Theorem \ref{thm:multiple_RIT}, we focus on reevaluating the most frequently occurring components of  $\mathrm{Var}\left( T_{K}\right)$ under the condition that $\zeta_{1,k} = 0$ for all $0 \leqslant k \leqslant K$:
\begin{equation*}
	\begin{aligned}
		\mathrm{Var}\left( T_{K}\right)  &= \left\lbrace \prod_{j=0}^K \binom{n_j}{m_j}^{-2}\right\rbrace \left[  \sum_{k=1}^K \binom{n_k}{2}\binom{n_k-2}{m_k -2}\binom{n_k-m_k}{m_k -2} \prod_{l \neq k} \left\lbrace \binom{n_l}{m_l}\binom{n_l-m_l}{m_l}\right\rbrace  \zeta_{2,k} \right.   \\
	\end{aligned}
\end{equation*}
\begin{equation*}
	\begin{aligned}
		& \qquad + \sum_{1\leqslant k_1 < k_2 \leqslant K} \binom{n_{k_1}}{1}\binom{n_{k_1}-1}{m_{k_1} -1}\binom{n_{k_1}-m_{k_1}}{m_{k_1} -1} \binom{n_{k_2}}{1}\binom{n_{k_2}-1}{m_{k_2} -1}\binom{n_{k_2}-m_l}{m_{k_2} -1} \\
		& \qquad \qquad \qquad \qquad \qquad  \left. \prod_{l \neq (k_1,k_2)} \left\lbrace \binom{n_l}{m_l}\binom{n_l-m_l}{m_l}\right\rbrace  \zeta_{2,k_1,k_2} \right] + o\left(n_1^{-2} \right) \\
		&\rightarrow \sum_{k=1}^K \frac{m_k^2 (m_k-1)^2}{2 n_k^2} \zeta_{2,k} + \sum_{1\leqslant k_1 < k_2 \leqslant K} \frac{m_{k_1}^2 m_{k_2}^2}{n_{k_1} n_{k_2}} \zeta_{2,k_1,k_2}.
	\end{aligned}
\end{equation*}
Together with assumption $n_k/n_1 \rightarrow r_k$, we can proof the second part of Theorem \ref{thm:multiple_RIT}.

%
%
%
%

\scsubsection{Appendix C.6. Discussion on Selecting $s$}

{
	We first provide the expressions for the power of RIT and BIT under the given alternative hypothesis $\mbox{H}_1$. For the first order RIT, assuming that $\mbox{H}_1$ holds and $\mathrm{E}T=\mu_0 \neq 0$, we have \(n_1^{1/2} (T - \mu_0) \xrightarrow{d} N(0, m_1^2 \xi_{0,1})\) and \(n_1^{1/2} (T_{\rm S} - \mu_0) \xrightarrow{d} N(0, m_1^2 \xi_{0,1} + m_0^2 \xi_{1,0} / s)\). The rejection region of $T$ with size $\alpha$ is given by $\{T < \Phi^{-1}(\alpha/2) (m_1^2 \xi_{0,1}/n_1)^{1/2} \} \cup \{T > \Phi^{-1}(1-\alpha/2) (m_1^2 \xi_{0,1}/n_1)^{1/2} \}$. And that for BIT $T_{\rm S}$ could be similarly defined.
	\begin{proposition} (The power of $T$ and $T_{\rm S}$)
		\label{thm: power}
		Assume the alternative hypothesis $\mbox{H}_1$ is true with $\mathrm{E}T=\mu_0 \neq 0$ and $n_1/n_0\rightarrow 0$. Then we have, (1) the power of RIT is
		$1+\Phi \{\Phi^{-1}(\alpha/2) - \mu_0 {( m_1^2 \xi_{0,1}/n_1 )^{-1/2}}\}  - \Phi\{\Phi^{-1}(1-\alpha/2) -\mu_0 {( m_1^2 \xi_{0,1}/n_1)^{-1/2}}\}$;
		(2) the power of BIT is $1+\Phi \{ \Phi^{-1}(\alpha/2) -  \mu_0 (m_1^2 \xi_{0,1}/n_1 + m_0^2 \xi_{1,0}/s n_1)^{-1/2} \} - \Phi \{\Phi^{-1}(1-\alpha/2) - \mu_0 (m_1^2 \xi_{0,1}/n_1 + m_0^2 \xi_{1,0}/s n_1)^{-1/2} \}$.
	\end{proposition}
	\begin{proof}[Proof of Proposition \ref{thm: power}]
		To obtain the power of \(T\) and \(T_{\rm S}\) under the alternative hypothesis, we only need to calculate their corresponding rejection regions and the probabilities of statistics falling into rejection regions. By Theorem \ref{thm:gen_U} and Theorem \ref{thm:sample_U}, the rejection region for \(T\) with size \(\alpha\) is given by \(\{T < \Phi^{-1}(\alpha/2) (m_1^2 \xi_{0,1}/n_1)^{1/2} \} \cup \{T > \Phi^{-1}(1-\alpha/2) (m_1^2 \xi_{0,1}/n_1)^{1/2} \}\), and the rejection region for \(T_{\rm S}\) is $\{T_{\rm S} < \Phi(\alpha/2) (m_1^2 \xi_{0,1}/n_1 + m_0^2 \xi_{1,0} / sn_1)^{1/2}\} \cup \{T_{\rm S} > \Phi(1-\alpha/2) (m_1^2 \xi_{0,1}/n_1 + m_0^2 \xi_{1,0} / sn_1)^{1/2}\}$.
		
		The probabilities of \(T\) falling into left rejection region can be written by $\mathrm{P} \{ T <  \Phi^{-1}(\alpha/2) (m_1^2 \xi_{0,1}/n_1)^{1/2} \} = \mathrm{P} \{ (T - \mu_0)  (m_1^2 \xi_{0,1}/n_1)^{-1/2} <  \Phi^{-1}(\alpha/2) - \mu_0 (m_1^2 \xi_{0,1}/n_1)^{-1/2} \} =  \mathrm{P} \{ Z <  \Phi^{-1}(\alpha/2) - \mu_0 (m_1^2 \xi_{0,1}/n_1)^{-1/2} \}  =  \Phi\{\Phi^{-1}(\alpha/2) - \mu_0 (m_1^2 \xi_{0,1}/n_1)^{-1/2} \} $, where $Z \sim N(0,1)$. Same technique can be applied to conclude $\mathrm{P} \{ T >  \Phi^{-1}(1-\alpha/2) (m_1^2 \xi_{0,1}/n_1)^{1/2} \} = 1 -  \Phi\{\Phi^{-1}(1-\alpha/2) - \mu_0 (m_1^2 \xi_{0,1}/n_1)^{-1/2} \}$.
		
		The probabilities of $T_{\rm S}$ falling into left rejection region can be written by $\mathrm{P} \{T_{\rm S} <  \Phi^{-1}(\alpha/2) (m_1^2 \xi_{0,1}/n_1 + m_0^2 \xi_{1,0} / sn_1)^{1/2} \} = \mathrm{P} \{ (T_{\rm S} - \mu_0)  (m_1^2 \xi_{0,1}/n_1 + m_0^2 \xi_{1,0} / sn_1)^{-1/2} <  \Phi^{-1}(\alpha/2) - \mu_0 (m_1^2 \xi_{0,1}/n_1 + m_0^2 \xi_{1,0} / sn_1)^{-1/2} \} =  \mathrm{P} \{ Z <  \Phi^{-1}(\alpha/2) - \mu_0 (m_1^2 \xi_{0,1}/n_1 + m_0^2 \xi_{1,0} / sn_1)^{-1/2} \}  =  \Phi\{\Phi^{-1}(\alpha/2) - \mu_0 (m_1^2 \xi_{0,1}/n_1 + m_0^2 \xi_{1,0} / sn_1)^{-1/2} \} $, where $Z \sim N(0,1)$. Same technique can be applied and $\mathrm{P} \{ T_{\rm S} >  \Phi^{-1}(1-\alpha/2) (m_1^2 \xi_{0,1}/n_1 + m_0^2 \xi_{1,0} / sn_1)^{1/2} \} = 1 -  \Phi\{\Phi^{-1}(1-\alpha/2) - \mu_0 (m_1^2 \xi_{0,1}/n_1 + m_0^2 \xi_{1,0} / sn_1)^{-1/2} \}$.
	\end{proof}
	
	\noindent
	It is evident from Proposition \ref{thm: power} that the power of \( T \) and \( T_s \) are both driven by \( n_1 \) rather than $n$. The power of \( T_s \), depending on \( \mu_0 \left( m_1^2 \xi_{0,1} / n_1 + m_0^2 \xi_{1,0} / s n_1 \right)^{-1/2} \), is slightly lower than that of \( T \), which depends on \(\mu_0 \left( m_1^2 \xi_{0,1}/n_1 \right)^{-1/2} \). Similarly, as concluded in Theorem \ref{thm:sample_U}, from the perspective of power analysis, the BIT significantly reduces computational complexity while maintaining the same convergence rate as RIT by excluding the majority of controls.
	
	For the second order RIT, we only consider the case of high-dimensionality, since we could not obtain the explicit asymptotic distribution for fixed dimension. For the high-dimensional hypothesis testing statistic \(T_{\infty}\), the dimension $p$ of $X$ will change, so we assume that $\mathrm{E}h(\cdot)$ has a limit and redefine $\mu_0$ as $\lim_{n_1 \to \infty} \mathrm{E}h(\cdot) = \mu_0 \neq 0$. It could be verified that the asymptotic distribution of $T_{\infty}$ is $n_1 (T_{\infty} - \mu_0) / \xi_{0,2}^{1/2} \xrightarrow{d} N(0, m_1^2 (m_1 - 1)^2)$ under the alternative hypothesis $\mbox{H}_1$.  We have the following conclusion. The verification is similar and thus omitted.
	\begin{proposition} (The power of $T_{\infty}$)
		\label{cor:power T_div_p}
		Assume the alternative hypothesis \( \mbox{H}_1 \) is true with \( \mathrm{E} T_{\infty} = \mu_0 \neq 0 \), \( n_1/n_0 \to 0 \), \( \xi_{0,1} = \xi_{1,0} = 0 \), \( p \to \infty \), and \( h_{0,2}(\cdot, \cdot) \) satisfying condition \eqref{eq:condition_clt}. Then the power of \( T_{\infty} \) is $\Phi \{\Phi^{-1}(\alpha/2) -  \mu_0 n_1 / m_1 (m_1 - 1) \xi_{0,2}^{1/2}\} + 1 - \Phi \{\Phi^{-1}(1-\alpha/2) - \mu_0 n_1 / m_1 (m_1 - 1) \xi_{0,2}^{1/2}\}$.
		As \( s \to \infty \), the statistic \( T_{\rm S} \) can achieve the same power as \( T_{\infty} \).
\end{proposition}}

\scsubsection{Appendix C.7. Local Power Analysis}

{Local power analysis follows a similar approach to that described in Appendix C.6. For the first-order RIT, $\mathrm{E}h(\cdot)$ will vary with the local alternatives. We can verify that under $\mbox{H}_{1,n_1}(\Delta_{0})$: $\Delta_{n_1} = n_1^{-1/2} \Delta_0$, $\mathrm{E}h(\cdot) =\sum_{k=1}^{m_1} \binom{k}{m_1} {\Delta_{n_1}}^k  \mu_{G,k} = m_1 \Delta_0 \mu_{G,1}/n_1^{1/2} + O(n_1^{-1})$, which indicates $\lim_{n_1 \to \infty} n_1^{1/2} \mathrm{E}T = m_1 \Delta_0 \mu_{G,1}$. Thus, we obtain \(n_1^{1/2} T  \xrightarrow{d} N(m_1 \Delta_0 \mu_{G,1}, m_1^2 \xi_{0,1})\). Correspondingly, under the null hypothesis, we have $n_1^{1/2} T \xrightarrow{d} N(0, m_1^2 \xi_{0,1})$, where there exists a shift term between the two asymptotic distributions. The rejection region of $T$ with size $\alpha$ is given by $\{n_1^{1/2}T < \Phi^{-1}(\alpha/2) m_1 \xi_{0,1}^{1/2} \} \cup \{n_1^{1/2}T > \Phi^{-1}(1-\alpha/2)m_1 \xi_{0,1}^{1/2} \}$. Letting $C = \{\Phi^{-1}(1 - \alpha/2) - \Phi^{-1}(1 - \beta)\} \xi_{0,1}^{1/2}/|\mu_{G,1}|$ and $\mu_{G,1}>0$ without loss of generality, it can be verified for any $\Delta_{0} > C$,
	\begin{equation*}
		\begin{aligned}
			& \lim_{n_1 \to \infty} \mathrm{P}\left\lbrace \phi_T(T) = 1 \mid \mbox{H}_{1,n_1}(\Delta_{0})  \right\rbrace  = \lim_{n_1 \to \infty} \mathrm{P} \left\lbrace n_1^{1/2}T > \Phi^{-1}(1-\alpha/2)m_1 \xi_{0,1}^{1/2} \right\rbrace \\
			= & \lim_{n_1 \to \infty} \mathrm{P} \left\lbrace \frac{n_1^{1/2}T - m_1 \Delta_0 \mu_{G,1}}{m_1 \xi_{0,1}^{1/2}}  > \Phi^{-1}(1-\alpha/2) - \frac{\Delta_0 \mu_{G,1}}{\xi_{0,1}^{1/2}} \right\rbrace
		\end{aligned}
	\end{equation*}
	\begin{equation*}
		\begin{aligned}
			= & \mathrm{P} \left\lbrace Z > \Phi^{-1}(1-\alpha/2) - \frac{\Delta_0 \mu_{G,1}}{\xi_{0,1}^{1/2}} \right\rbrace \\
			> & \mathrm{P} \left\lbrace Z > \Phi^{-1}(1-\alpha/2) - \Phi^{-1}(1-\alpha/2) + \Phi^{-1}(1-\beta) \right\rbrace = 1 - \beta.
		\end{aligned}
	\end{equation*}
	And for any $\Delta_{0} < -C$, we have
	\begin{equation*}
		\begin{aligned}
			& \lim_{n_1 \to \infty} \mathrm{P}\left\lbrace \phi_T(T) = 1 \mid \mbox{H}_{1,n_1}(\Delta_{0})  \right\rbrace  = \lim_{n_1 \to \infty} \mathrm{P} \left\lbrace n_1^{1/2}T < \Phi^{-1}(\alpha/2)m_1 \xi_{0,1}^{1/2} \right\rbrace \\
			= & \lim_{n_1 \to \infty} \mathrm{P} \left\lbrace \frac{n_1^{1/2}T - m_1 \Delta_0 \mu_{G,1}}{m_1 \xi_{0,1}^{1/2}}  < \Phi^{-1}(\alpha/2) - \frac{\Delta_0 \mu_{G,1}}{\xi_{0,1}^{1/2}} \right\rbrace \\
			> & \mathrm{P} \left\lbrace Z < \Phi^{-1}(\alpha/2) + \Phi^{-1}(1-\alpha/2) - \Phi^{-1}(1-\beta) \right\rbrace = 1 - \beta.
		\end{aligned}
	\end{equation*}

	The same techniques can be applied to the remaining local power analysis, and therefore we omit them here.
}

\scsubsection{Appendix D. Supplementary Tables}

In this subsection, we provide the results for the second-order RIT and BIT in the simulation for Example 2 in Table \ref{tab:computation_dcor_and_ipcor}, and the feature definitions for the Credit Dataset in Table \ref{tab:credit_intro}.

\begin{table}[htb]
	\begin{center}
		\caption{Comparison of Empirical Size and Power of Rescaled Distance Covariance $T_{\rm dcov}$ and  Rescaled Projection Correlation $T_{\rm IPcov}$ and the corresponding BITs with Different $n_1 s$.}
		\label{tab:computation_dcor_and_ipcor}
		\renewcommand\arraystretch{1.5} 
		\begin{tabular}{l|ccccc|ccccc}
			\hline \hline
			&\multicolumn{5}{c|}{Empirical Size}&\multicolumn{5}{c}{Empirical Power}\\
			\hline
			&\multicolumn{1}{c|}{}&\multicolumn{4}{c|}{$n_1  s$}	&\multicolumn{1}{c|}{}&\multicolumn{4}{c}{$n_1  s$}\\
			\hline
			&\multicolumn{1}{c|}{ RIT}&1000&1500&2000&2500&\multicolumn{1}{c|}{ RIT}&1000&1500&2000&2500\\
			\hline
			$n_1$&\multicolumn{10}{c}{Rescaled Distance Covariance $T_{\rm dcov}$ Eg.1}\\
			30 & 0.052 & 0.057 & 0.059 & 0.063 & 0.052 & 0.843 & 0.822 & 0.835 & 0.841 & 0.838 \\
			50 & 0.053 & 0.059 & 0.057 & 0.068 & 0.052 & 0.983 & 0.971 & 0.980 & 0.982 & 0.981 \\
			70 & 0.050 & 0.069 & 0.073 & 0.069 & 0.059 & 0.999 & 0.999 & 0.998 & 0.999 & 0.999 \\
			100 & 0.057 & 0.068 & 0.064 & 0.056 & 0.061 & 1.000 & 0.999 & 1.000 & 1.000 & 1.000 \\
			\hline
			$n_1$&\multicolumn{10}{c}{Rescaled Projection Covariance $T_{\rm IPcov}$ Eg.1}\\
			30 & 0.053 & 0.060 & 0.050 & 0.056 & 0.049 & 0.826 & 0.808 & 0.816 & 0.821 & 0.820 \\
			50 & 0.049 & 0.058 & 0.044 & 0.063 & 0.058 & 0.988 & 0.975 & 0.981 & 0.976 & 0.986 \\
			70 & 0.064 & 0.062 & 0.064 & 0.061 & 0.055 & 1.000 & 0.994 & 0.998 & 0.999 & 1.000 \\
			100 & 0.062 & 0.070 & 0.076 & 0.063 & 0.060 & 1.000 & 1.000 & 1.000 & 1.000 & 1.000 \\
			\hline
			$n_1$&\multicolumn{10}{c}{Rescaled Distance Covariance $T_{\rm dcov}$ Eg.2}\\
			30 & 0.053 & 0.048 & 0.068 & 0.062 & 0.044 & 0.799 & 0.766 & 0.789 & 0.802 & 0.794 \\
			50 & 0.063 & 0.057 & 0.066 & 0.047 & 0.058 & 0.969 & 0.952 & 0.973 & 0.962 & 0.967 \\
			70 & 0.051 & 0.072 & 0.064 & 0.059 & 0.054 & 0.996 & 0.990 & 0.996 & 0.997 & 0.996 \\
			100 & 0.062 & 0.083 & 0.072 & 0.063 & 0.059 & 1.000 & 1.000 & 0.999 & 1.000 & 1.000 \\
			\hline
			$n_1$&\multicolumn{10}{c}{Rescaled Projection Covariance $T_{\rm IPcov}$ Eg.2}\\
			30 & 0.047 & 0.066 & 0.052 & 0.062 & 0.054 & 0.805 & 0.782 & 0.786 & 0.798 & 0.803 \\
			50 & 0.052 & 0.068 & 0.061 & 0.060 & 0.053 & 0.977 & 0.962 & 0.973 & 0.969 & 0.972 \\
			70 & 0.048 & 0.075 & 0.064 & 0.067 & 0.059 & 1.000 & 0.992 & 0.997 & 1.000 & 1.000 \\
			100 & 0.061 & 0.074 & 0.077 & 0.059 & 0.068 & 1.000 & 0.998 & 1.000 & 1.000 & 1.000 \\
			\hline
		\end{tabular}
	\end{center}
\end{table}

\begin{table}[ht]
	\centering
	\renewcommand\arraystretch{1}
	\caption{Feature Groups and Detailed Feature Definitions in the Credit Dataset}
	\label{tab:credit_intro}
	\begin{tabular}{p{5cm}|p{8cm}}
		\hline
		\textbf{Group Name}& \textbf{Feature Definitions} \\
		\hline
		{\textbf{Repayment Ability}} & Total balance on credit cards and personal lines of credit except real estate; Monthly income; Monthly debt payments divided by monthly gross income; Number of dependents in the family \\
		\hline
		{\textbf{Credit History}} & Number of times borrower has been 30-59 days past due but no worse in the last 2 years; Number of times borrower has been 60-89 days past due but no worse in the last 2 years; Number of times borrower has been 90 days or more past due \\
		\hline
		{\textbf{Property Status}} & Number of Open loans (installment like car loan or mortgage) and Lines of credit (e.g. credit cards); Number of mortgage and real estate loans including home equity lines of credit; Age of borrower in years \\
		\hline
	\end{tabular}
	
\end{table}

\bibliographystyle{asa}
\bibliography{Mybib}

\newpage

\end{document}